\documentclass[11pt,reqno]{amsart}

\usepackage[
  letterpaper,
  margin=1in
]{geometry}

\usepackage{tikz}
\usepackage{secdot}
\usepackage{titlesec}
\titleformat{\subsection}
  {\normalfont\bfseries}
  {\thesubsection.}
  {0.6em}
  {}
  [\vspace{0.4\baselineskip}]

\usepackage{amsmath,amssymb,amsthm,mathtools}
\DeclareMathOperator{\Det}{Det}
\DeclareMathOperator{\Tr}{Tr}

\usepackage{newtxtext,newtxmath}
\usepackage{bm}

\makeatletter
\tagsleft@false
\makeatother

\usepackage{microtype}

\usepackage{graphicx}
\usepackage{subcaption}
\usepackage{booktabs}
\usepackage{float}
\usepackage[section]{placeins}
\usepackage{tabularx}
\usepackage{multirow}
\usepackage{array}

\newcolumntype{L}{>{\raggedright\arraybackslash}X}

\newcommand{\cond}[1]{\ensuremath{#1}}

\usepackage[hidelinks]{hyperref}
\hypersetup{
  colorlinks=true,
  linkcolor=blue,
  citecolor=blue,
  urlcolor=blue
}

\usepackage[nameinlink,noabbrev]{cleveref}

\usepackage{comment}

\theoremstyle{plain}
\newtheorem{theorem}{Theorem}[section]
\newtheorem{proposition}[theorem]{Proposition}
\newtheorem{lemma}[theorem]{Lemma}

\theoremstyle{definition}

\theoremstyle{remark}

\title[Harvesting and Group Defense]{Dynamics of a Harvested Predator-Prey Model with a Generalized Holling Type~III Response Capturing Group Defense}

\makeatletter
\def\@setauthors{\begingroup
  \def\thanks{\protect\thanks@warning}  \trivlist \centering\footnotesize \@topsep30\p@\relax
  \item\relax
  \andify\authors
  {\normalsize\normalfont \authors\par}  \endtrivlist
\endgroup}
\makeatother

\title[Harvesting and Group Defense]{Dynamics of a Harvested Predator-Prey Model with a Generalized Holling Type~III Response Capturing Group Defense}

\author{Wael El Khateeb\textsuperscript{1,2}}
\thanks{\textsuperscript{1} Department of Mathematics, Emirates Aviation University, Dubai, UAE.
Email: \href{mailto:eaus559@eau.ac.ae}{eaus559@eau.ac.ae}.}

\thanks{\textsuperscript{2} Department of Mathematics and Statistics, University of Toledo, Toledo, OH, USA.
Email: \href{mailto:wael.elkhateeb@rockets.utoledo.edu}{wael.elkhateeb@rockets.utoledo.edu}.}
\author{Farjana Akter\textsuperscript{3}}
\thanks{\textsuperscript{3} Department of Mathematics and Statistics, University of Toledo, Toledo, OH, USA.
Email: \href{mailto:Farjana.Akter@rockets.utoledo.edu}{Farjana.Akter@rockets.utoledo.edu}.}

\author{Chanaka Kottegoda\textsuperscript{4}}
\thanks{\textsuperscript{4} Department of Mathematics and Physics, Marshall University, Huntington, WV, USA.
Email: \href{mailto:kottegoda@marshall.edu}{kottegoda@marshall.edu}.}

\author{Muner Mustafa Abou Hasan\textsuperscript{5}}
\thanks{\textsuperscript{5} Department of Mathematics, Emirates Aviation University, Dubai, UAE.
Email: \href{mailto:muner.abouhasan@emirates.com}{muner.abouhasan@emirates.com}.}

\date{}

\begin{document}
\maketitle

\begin{abstract}
We investigate a predator-prey model with constant prey harvesting and a generalized Holling type III functional response, focusing on the biological and dynamical consequences of allowing the parameter $b$ to be negative. In this regime, predation efficiency declines at high prey densities, potentially reflecting group defense. This change permits an additional coexistence equilibrium and introduces new bifurcation mechanisms affecting population persistence. We establish the existence of a codimension-three Bogdanov--Takens bifurcation, derive a codimension-three nilpotent-saddle unfolding, and investigate a codimension-two generalized Hopf bifurcation and the emergence of multiple limit cycles. Our analysis show that moderate harvesting weakens the protection conferred by high prey abundance; excessive harvesting leaves too few prey to sustain predator growth; and low harvesting preserves strong group defense, requiring high predator density for population persistence.
\end{abstract}

\noindent\textbf{Keywords:} predator-prey model; mathematical ecology; population dynamics; harvesting; bifurcation theory\\
\textbf{MSC 2020:} 34C23, 92D25, 92D40, 34C05, 34C07, 37G10, 37G15.

\section{Introduction}

The Lotka--Volterra model originated with Lotka's work on autocatalytic reactions \cite{Lotka1910}, was adapted by Lotka to predator-prey interactions \cite{lotka1925elements}, and was independently applied by Volterra to Adriatic fish dynamics \cite{volterra1926fluctuations}. Its classical formulation describes prey growth, predation, conversion of prey into predator biomass, and predator mortality through ordinary differential equations. Later models introduced Holling functional responses to account for feeding saturation and prey-dependent intake rates \cite{holling1959components}, as well as spatial structure and delays to describe movement and the time required for predator recruitment \cite{hupfaker1958experimental}. For example, El Khateeb et al.\ \cite{khateeb2025globalhopfbifurcationconnected} studied local and global Hopf bifurcations in a delayed predator-prey model and showed that delays can permit oscillatory coexistence at lower carrying capacities than in the corresponding ODE model. These developments show how biological mechanisms can change both the conditions for coexistence and the onset of population oscillations.\\

Harvesting and group defense can have competing effects on predator-prey interactions. Constant prey harvesting reduces the food available to predators and can threaten the persistence of both populations \cite{Clark1974}. Prey that form defensive groups, however, may become less vulnerable to predation as their abundance increases \cite{geritz2013group}. A non-monotone functional response can capture this reduction in feeding efficiency at high prey densities \cite{lamontagne2008bifurcation,huang2014bifurcations}. In this setting, removing prey changes not only their abundance but also the protection associated with that abundance. The effect of harvesting on predator establishment therefore depends on whether prey density lies below or above the level at which predation is most effective. This interaction motivates a closer study of harvesting in models that capture a group defense strategy.\\

A starting point is the harvested predator-prey system studied by Etoua and Rousseau \cite{etoua2010bifurcation}:

\begin{equation}\label{Original Model}
    \begin{cases}
         \frac{dx}{dt}=r x(1-\frac{x}{K})-myp(x)-h,\\
         \frac{dy}{dt}=y(-d+cmp(x)),
    \end{cases}
\end{equation}
where $x$ and $y$ are the prey and predator populations, respectively, $K$ is the carrying capacity, $r$ is the growth rate of the prey, $h$ is the constant prey-harvesting rate, $d$ is the natural per capita death rate of predators, $m$ is a positive predation-rate coefficient, and $c$ is the conversion efficiency. The functional response is taken as a Holling type III $p(x)=\frac{x^2}{a x^2+b x+1}$, where $a>0$ is the inverse of the satiation limit and $b\geq0$ is a constant. System \eqref{Original Model} exhibits at most three equilibria, at most two of which are on the $x$-axis. The nilpotent equilibrium of multiplicity three arises on the invariant line and, while its formal codimension can reach four, for $b>0$, the presence of this invariant line reduces the effective codimension of this degeneracy to three. Working with normal forms that preserve the invariant line, one obtains a systematic description of the nilpotent saddle bifurcation and its adjacent transition sets. The bifurcation diagram of a nilpotent saddle includes attracting and repelling Hopf bifurcations, a codimension-two Hopf bifurcation, attracting and repelling heteroclinic-loop bifurcations, a codimension-two heteroclinic loop, as well as double limit cycles. The center and reversibility cases at $b=0$ were subsequently resolved by Laurin and Rousseau \cite{laurin2011organizing}.\\

For $-2\sqrt a<b<0$, the functional response is non-monotone: it increases at low prey density, reaches a maximum at $x=-\frac{2}{b}$, and then decreases toward $\frac{1}{a}$ as prey density grows (Fig.\ \ref{holling_types_a},\ref{holling_types_b}) \cite{lamontagne2008bifurcation,huang2014bifurcations}. The decrease in intake at high prey densities provides a way to represent group defense while retaining a bounded feeding rate. The corresponding model without harvesting, including its codimension-three Bogdanov--Takens bifurcation, was studied in \cite{lamontagne2008bifurcation}.\\

In this paper, we extend the analysis of \cite{etoua2010bifurcation} to this negative-$b$ regime. We determine when two coexistence equilibria are biologically feasible and study their collisions with each other and with the predator-free equilibria. The analysis establishes a codimension-three Bogdanov--Takens bifurcation, a codimension-three nilpotent-saddle unfolding, and a codimension-two generalized Hopf bifurcation that can generate two small-amplitude limit cycles. We then connect the boundary bifurcations to predator invasion and identify an intermediate harvesting interval in which a rare predator population can establish. The comparison with the positive-$b$ regime is summarized in Table~\ref{tab:differences between b>0 and b<0}.\\

\begin{figure}[!htbp]
  \centering
  \begin{subfigure}{0.45\linewidth}
    \centering
    \begin{tikzpicture}
      \node[inner sep=0] (img) {
        \includegraphics[width=\linewidth]{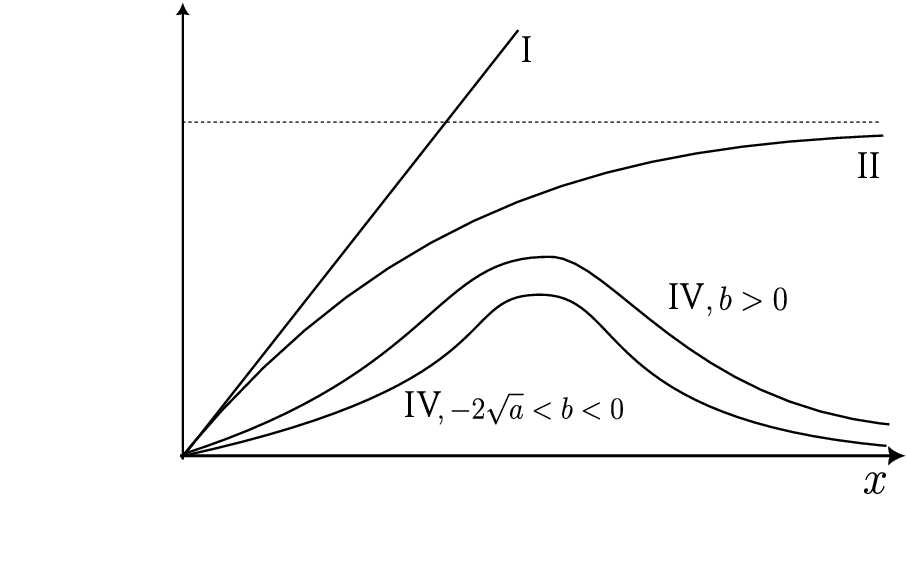}
      };
      \node[anchor=south west, align=left, font=\small]
        at ([xshift=10mm,yshift=1mm]img.north west)
        {Number of prey\\consumed};
    \end{tikzpicture}
    \caption{}
    \label{holling_types_a}
  \end{subfigure}
  \hfill
  \begin{subfigure}{0.45\linewidth}
    \centering
    \begin{tikzpicture}
      \node[inner sep=0] (img) {
        \includegraphics[width=\linewidth]{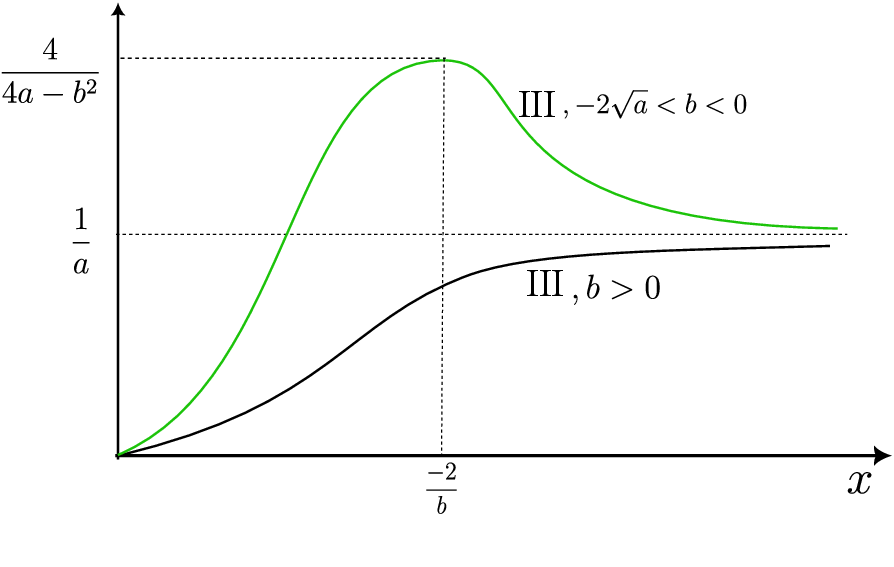}
      };
      \node[anchor=south west, align=left, font=\small]
        at ([xshift=5mm,yshift=1mm]img.north west)
        {Number of prey\\consumed};
    \end{tikzpicture}
    \caption{}
    \label{holling_types_b}
  \end{subfigure}
  \caption{Holling Types I, II, III and IV}
  \label{holling_types}
\end{figure}

This paper is organized as follows:
\begin{itemize}
\item \textbf{Section 2.} We nondimensionalize the model, determine the conditions for equilibria to exist, and establish boundedness for solutions that remain in the first quadrant.

\item \textbf{Section 3.} We classify the equilibria and study their bifurcations, including the codimension-2 generalized Hopf bifurcation, the codimension-3 Bogdanov--Takens bifurcation (BT-3), and the codimension-3 nilpotent-saddle bifurcation (NS-3). We derive versal unfoldings of the latter two and examine collisions with the predator-free equilibria.

\item \textbf{Section 4.} We discuss the biological implications, with particular attention to how moderate harvesting can facilitate predator establishment when prey exhibit group defense.
\end{itemize}

\section{Qualitative Dynamics of the ODE Model}

Let
\(
\Omega=\{(x,y)\in\mathbb{R}^2:\ x\ge0,\ y\ge0\}.
\)
By applying the nondimensional rescaling
\begin{equation}
  (\bar{x},\bar{y},\bar{t},\bar{h},\bar{a},\bar{b},\bar{K},\bar{d})
  \to
  \Bigl(
  \sqrt{\frac{cm}{r}}x,\
  \sqrt{\frac{m}{cr}}y,\
  rt,\
  \sqrt{\frac{cm}{r}}\frac{h}{r},\
  \frac{ar}{cm},\
  \frac{b\sqrt{r}}{\sqrt{cm}},\
  \sqrt{\frac{cm}{r}}K,\
  \frac{d}{r}
  \Bigr),
\end{equation}
Omitting the bars from the rescaled variables and parameters, and writing $p$ for the rescaled functional response, system~\eqref{Original Model} takes the planar form
\begin{equation}\label{Rescaled Main Model}
    \begin{cases}
        \dot{x}=x\bigl(1-\frac{x}{K}\bigr)-y\,p(x)-h,\\
        \dot{y}=y\bigl(-d+p(x)\bigr).
    \end{cases}
\end{equation}
where
\[
p(x)=\frac{x^2}{ax^2+bx+1}.
\]
Throughout this paper, we assume $a,d,K,h>0$ and $-2\sqrt a<b<0$. The equilibria of \eqref{Rescaled Main Model}, together with their existence conditions in $\Omega$, are summarized as follows:\\
\begin{itemize}
\item The predator-free equilibria are
$
E_1=\Bigl(x_{b_1}=\frac{K-\sqrt{K^2-4hK}}{2},\,0\Bigr)\text{ and }
E_2=\Bigl(x_{b_2}=\frac{K+\sqrt{K^2-4hK}}{2},\,0\Bigr),
$
which exist for $K\ge4h$.

\item The coexistence equilibrium
$
E_3=\Bigl(
x_\ell=\frac{-bd-\sqrt{b^2d^2+4d(1-ad)}}{2(ad-1)},\
y_\ell=\frac{x_\ell\bigl(1-\frac{x_\ell}{K}\bigr)-h}{d}
\Bigr)\in\Omega
$
for $0<d<\frac{4}{4a-b^2}$, $d\ne\frac{1}{a}$, $K>4h$, and $x_{b_1}<x_\ell<x_{b_2}$. For $d=\frac{1}{a}$, $x_\ell=-\frac{1}{b}$, and $E_3$ lies in the open first quadrant when $K>4h$ and $x_{b_1}<x_\ell<x_{b_2}$.

\item The coexistence equilibrium
$
E_4=\Bigl(
x_r=\frac{-bd+\sqrt{b^2d^2+4d(1-ad)}}{2(ad-1)},\
y_r=\frac{x_r\bigl(1-\frac{x_r}{K}\bigr)-h}{d}
\Bigr)\in\Omega
$
for $\frac{4}{4a-b^2}>d>\frac{1}{a}$, $K>4h$, and $x_{b_1}<x_r<x_{b_2}$.\\
\end{itemize}
For $\frac{4}{4a-b^2}>d>\frac{1}{a}$ with $x_{b_1}<x_\ell,x_r<x_{b_2}$, both $E_3$ and $E_4$ lie in the open first quadrant and therefore coexist in $\Omega$. If a coexistence equilibrium reaches $E_1$ or $E_2$, its predator coordinate becomes zero, yielding a boundary collision. In contrast, when $d<\frac{1}{a}$, the equilibrium $E_4$ has a negative $x$-coordinate and is not biologically relevant, so it is discarded.

The theorem below establishes boundedness for trajectories that remain in the first quadrant.

\begin{theorem}
Assume $K>4h$. There exist positive constants $\mathcal H,\mathcal P,$ and $\mathcal R$ such that every orbit starting in the open first quadrant either approaches a segment of the invariant $x$-axis, exits the first quadrant through the positive $y$-axis, or intersects the rectangle $[\mathcal H,\mathcal P]\times[0,\mathcal R]$. Moreover, every orbit that remains in $\Omega$ for all forward time is bounded (see Fig. \ref{fig:H,P,R}).
\end{theorem}

\begin{proof}
We consider only orbits in $\Omega$, as in \cite{etoua2010bifurcation}. The $x$-axis is invariant, and no orbit can cross it. For $h\ne0$, trajectories can cross the $y$-axis only from right to left. Define the vertical lines $x=\mathcal P$ and $x=\mathcal H$ such that $0<\mathcal H<x_{b_1}<x_{b_2}<\mathcal P$ and $\mathcal P>K$. If $x>\mathcal P$ or $0\le x<\mathcal H$, then
\[
\frac{dx}{dt}\le x\Bigl(1-\frac{x}{K}\Bigr)-h<0.
\]
Hence, orbits cross the lines $x=\mathcal P$ and $x=\mathcal H$ from right to left. On $[0,\mathcal H]$, the right-hand side is bounded above by a strictly negative constant. A trajectory entering this strip therefore leaves the first quadrant in finite time. For a trajectory starting in the open first quadrant, its predator coordinate remains positive at every finite time before exit, so this exit occurs through the positive $y$-axis.\\

\noindent To establish boundedness without imposing any additional conditions on the response roots, set $W=x+y$. As long as a trajectory remains in $\Omega$,
\[
\dot W=x-\frac{x^2}{K}-h-dy,
\qquad
\dot W+dW=(1+d)x-\frac{x^2}{K}-h
\le \frac{K(1+d)^2}{4}-h=:C.
\]
Since $K>4h$, we have $C>0$. Integration yields
\[
W(t)\le W(0)e^{-dt}+\frac{C}{d}(1-e^{-dt}).
\]
Thus every orbit that remains in $\Omega$ for all forward time is bounded. Choose $\mathcal R>\frac{C}{d}$. Every such orbit eventually satisfies $y<\mathcal R$ and $x<\mathcal P$, and it cannot enter $x\le\mathcal H$ without subsequently leaving $\Omega$. Consequently it intersects $[\mathcal H,\mathcal P]\times[0,\mathcal R]$. The rectangle is an intersection region and need not be positively invariant. The different cases are summarized in Fig. \ref{fig:H,P,R}. This completes the proof.
\end{proof}

\begin{figure}[!htbp]
  \centering
  \begin{subfigure}{0.35\linewidth}
    \centering
    \includegraphics[width=\linewidth]{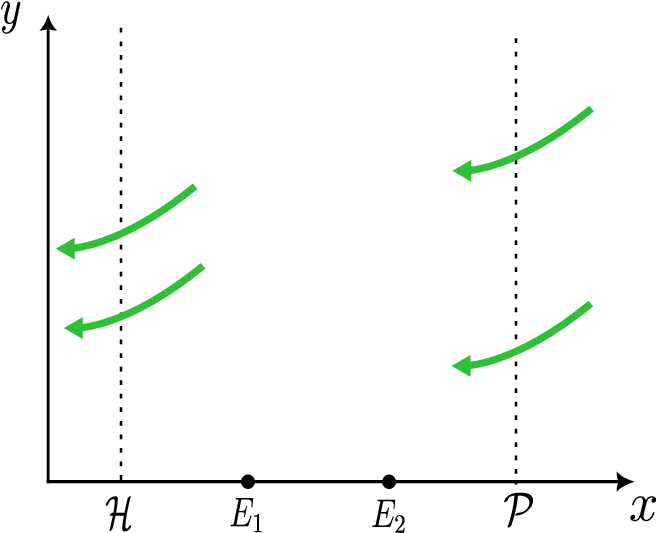}
    \caption{$d>p\!\left(\frac{-2}{b}\right)$}
    \label{fig:SvsXvsT}
  \end{subfigure}
  \hspace{3cm}
  \begin{subfigure}{0.35\linewidth}
    \centering
    \includegraphics[width=\linewidth]{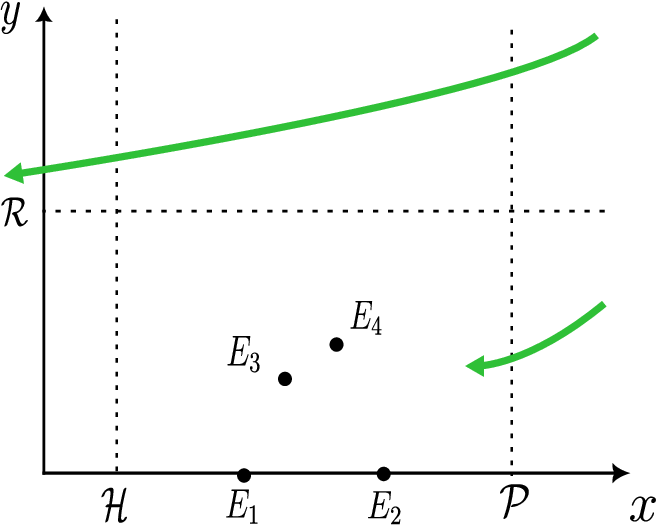}
    \caption{$d<p\!\left(\frac{-2}{b}\right)$}
    \label{fig:H,P,R}
  \end{subfigure}
  \caption{Orbits in the left panel either converge to an equilibrium point on the \(x\)-axis or leave the first quadrant through the positive \(y\)-axis. Orbits in the right panel either enter the rectangle \([\mathcal H,\mathcal P]\times[0,\mathcal R]\) from above or from the right, or leave the first quadrant through the positive \(y\)-axis.}
  \label{fig:orbits_boundedness}
\end{figure}

\section{High codimension bifurcations}

Having established the corresponding boundedness for trajectories in $\Omega$, we now turn to the qualitative behavior near the equilibrium points. In particular, we analyze the local dynamics in a sufficiently small neighborhood of each equilibrium by linearizing \eqref{Rescaled Main Model} and computing the associated Jacobian matrix:
\begin{equation}
J(x,y)=
\begin{bmatrix}
1-\frac{2x}{K}
-\frac{2yx}{ax^2+bx+1}
+\frac{yx^2(2ax+b)}{(ax^2+bx+1)^2}
&
-\frac{x^2}{ax^2+bx+1}
\\
y\left(
\frac{2x}{ax^2+bx+1}
-\frac{x^2(2ax+b)}{(ax^2+bx+1)^2}
\right)
&
-d+\frac{x^2}{ax^2+bx+1}
\end{bmatrix}.
\end{equation}
We denote the determinant and trace of $J(x,y)$ by $\Det(J(x,y))$ and $\Tr(J(x,y))$, respectively.
For $i=1,2,3,4$, let $\lambda_1^i$ and $\lambda_2^i$ denote the eigenvalues of $J(E_i)$. 
At a merged equilibrium $E_{jk}$ or $E_{jkl}$, we denote the two eigenvalues by $\lambda_1^{jk},\lambda_2^{jk}$ or $\lambda_1^{jkl},\lambda_2^{jkl}$, respectively, where $j,k,l\in\{1,2,3,4\}$.

According to the signs of $\Det(J(x,y))$ and $\Tr(J(x,y))$, and by examining higher-order terms in the degenerate cases, we classify the equilibria as hyperbolic saddles, sinks, sources, cusps, elliptic points, nilpotent saddles \cite{zhang1992qualitative,perko2013differential,hartman1964ode,kuznetsov2004elements,Chow_Li_Wang_1994,dumortier1987generic}. Each configuration is then studied separately, as it corresponds to a distinct local scenario arising in the bifurcations that occur when equilibria coalesce \cite{hirsch2013dedschaos,Chow_Li_Wang_1994}. The resulting classification is summarized in the discussion section (Table~\ref{bifurcation-summary}).

Here, an equilibrium is called an antisaddle when $\Det(J)>0$.

\begin{theorem}\label{E_3 antisaddle, E_4 saddle}
Assume that the distinct coexistence equilibria $E_3$ and $E_4$ exist in the open first quadrant. Then $E_3$ is of antisaddle type, whereas $E_4$ is of saddle type.
\end{theorem}

\begin{proof}
Evaluating the determinant $\Det(J(x,y))$ at $E_3=(x_\ell,y_\ell)$ and $E_4=(x_r,y_r)$ yields $\Det(J(E_3))=d\,y_\ell\,p'(x_\ell),$
$\Det(J(E_4))=d\,y_r\,p'(x_r).$ Since $d>0$ and $y_\ell,y_r>0$ for these interior equilibria, the sign of the determinant is determined by the sign of $p'(x)$. Since $p'(x)=\frac{x(bx+2)}{(ax^2+bx+1)^2}$ and $x_\ell<-\frac{2}{b}<x_r$, we have $p'(x_\ell)>0$ and $p'(x_r)<0$. Consequently, $\Det(J(E_3))>0$ and $\Det(J(E_4))<0$.
\end{proof}

When $E_3$ and $E_4$ are distinct and lie in the open first quadrant, their types are determined by the sign of the Jacobian determinant. We now consider the case in which these equilibria coincide. At $d=d_0=\frac{4}{4a-b^2}$, the two coexistence equilibria merge, and the Jacobian has a zero eigenvalue. The following theorem characterizes this saddle-node bifurcation and the stability of the resulting equilibrium.
\subsection{Codimension Three Bogdanov-Takens bifurcations }

\begin{theorem}
Let
\(
d=d_0=\frac{4}{4a-b^2},\)
and suppose
\(
K>-\frac{2}{b},
0<h<\frac{4+2bK}{-b^2K}.
\)
Then the coexistence equilibria $E_3$ and $E_4$ merge at
\[
E_{34}=(x_{lr},y_{lr}),\qquad
\text{where}\qquad
x_{lr}=-\frac{2}{b},\qquad
y_{lr}=\frac{x_{lr}(1-\frac{x_{lr}}{K})-h}{d_0}>0.
\]
If $K\ne K_0=-\frac{4}{b}$, then $E_{34}$ is a semi-hyperbolic saddle-node. Its nonzero eigenvalue is
\(
\lambda_2^{34}=\frac{bK+4}{bK}.
\)
Thus $\lambda_2^{34}<0$ for $-\frac{2}{b}<K<K_0$ and $\lambda_2^{34}>0$ for $K>K_0$. At $K=K_0$, the equilibrium has a double zero eigenvalue and is treated separately below.
\end{theorem}

\begin{proof}
The conditions
\[
C_2:\ 0<h<\frac{4+2bK}{-b^2K},\qquad
C_3:\ K>-\frac{2}{b},
\]
ensure that the $y$-coordinate of $E_{34}$ is positive. Evaluating the Jacobian matrix at $E_{34}$ yields
\begin{equation}
J(E_{34})=
\begin{bmatrix}
\frac{bK+4}{bK}&-\frac{4}{4a-b^2}\\
0&0
\end{bmatrix}.
\end{equation}
$J(E_{34})$ is singular, possessing a zero eigenvalue. For $K\ne K_0$, the corresponding eigenvalue-eigenvector pairs are given by:
\[
\left(
\lambda_1^{34}=0,\quad
v_1=
\begin{bmatrix}
\frac{4bK}{(-b^2+4a)(Kb+4)}\\
1
\end{bmatrix}
\right)
\quad\text{and}\quad
\left(
\lambda_2^{34}=\frac{Kb+4}{bK},\quad
v_2=
\begin{bmatrix}
1\\0
\end{bmatrix}
\right).
\]
To characterize the local phase portrait, we translate $E_{34}$ to the origin via $x_1=x-x_{lr},y_1=y-y_{lr}$. In the new coordinates $(x_1,y_1)$, the Taylor expansion of the vector field at the origin is
\begin{equation}\label{SN step 1}
\begin{cases}
\dot x_1=
\frac{bK+4}{bK}x_1
-\frac{4}{-b^2+4a}y_1
-\frac{b^4hK+2b^3K+16a}{4K(-b^2+4a)}x_1^2
+O(|x_1,y_1|^3),\\
\dot y_1=
\frac{(b^2hK+2bK+4)b^2}{4K(-b^2+4a)}x_1^2
+O(|x_1,y_1|^3).
\end{cases}
\end{equation}
Next, we simplify the linear part by putting it into Jordan form. This is achieved by the linear change of variables
\[
y_2=y_1,\qquad
x_2=x_1-\frac{4bK}{(-b^2+4a)(bK+4)}y_1.
\]
System (\ref{SN step 1}) becomes:
\begin{equation}\label{SN step 2}
\begin{cases}
\dot x_2=\frac{bK+4}{bK}x_2+O(|x_2,y_2|^2),\\
\dot y_2=
\frac{b^2(b^2hK+2bK+4)}{4K(-b^2+4a)}x_2^2
+\frac{2b^3(b^2hK+2bK+4)}{(-b^2+4a)^2(bK+4)}y_2x_2\\
\hspace{1.4cm}
+\frac{4b^4K(b^2hK+2bK+4)}{(-b^2+4a)^3(bK+4)^2}y_2^2
+O(|x_2,y_2|^3).
\end{cases}
\end{equation}
At this stage, the presence of a single zero eigenvalue implies the existence of a one-dimensional invariant center manifold $W^c(0)$ tangent to the eigenspace spanned by $v_1$. In particular, $W^c(0)$ admits a local graph representation
$W^c(0)=\{(x_2,y_2)\in R^2\mid x_2=\chi(y_2),\text{ for small }y_2\}.$ Restricting \eqref{SN step 2} to $W^c(0)$ yields the reduced dynamics
\begin{equation}
\begin{cases}
\dot x_2=\frac{bK+4}{bK}x_2+O(|x_2,y_2|^2),\\
\dot y_2=
\frac{4b^4K(b^2hK+2bK+4)}{(-b^2+4a)^3(bK+4)^2}y_2^2
+O(y_2^3).
\end{cases}
\end{equation}
The coefficient of $y_2^2$ in the second equation is nonzero under $C_2$, because the excluded value
\[
-\frac{2(bK+2)}{b^2K}
=\frac{4+2bK}{-b^2K}
\]
is precisely the boundary value for which $y_{lr}=0$. Consequently, for $K\ne K_0$, $E_{34}$ is a saddle-node equilibrium. Moreover, its transverse stability is determined by the sign of the linear coefficient in the $x_2$-equation of \eqref{SN step 2}, namely $\frac{bK+4}{bK}$. If $K>K_0=-\frac{4}{b}$, then $\lambda_2^{34}>0$, whereas if $-\frac{2}{b}<K<K_0$, then $\lambda_2^{34}<0$. At $K=K_0$ both eigenvalues vanish, so the semi-hyperbolic reduction no longer applies.\\
\end{proof}

We recall the normal-form reduction in Proposition~5.3 of \cite{lamontagne2008bifurcation}:
\begin{proposition}\label{propostion from Reousou 2008}
Let $a_{ij}$ be real constants for $i,j\in\mathbb N$, where $\mathbb N=\{0,1,2,\ldots\}$. The system
\[
\dot x=y,
\]
\[
\dot y=x^2+a_{30}x^3+a_{40}x^4
+y(a_{21}x^2+a_{31}x^3)
+y^2(a_{12}x+a_{22}x^2)
+O(|x,y|^5),
\]
is equivalent to
\[
X'=Y,
\]
\[
Y'=X^2+EX^3Y+O(|X,Y|^5),
\]
where $E=a_{31}-a_{30}a_{21}$.
\end{proposition}
The preceding proposition is obtained via a sequence of near-identity coordinate changes and a time rescaling, and it provides a convenient normal form reduction for degenerate equilibria. We now apply this framework to the coalesced equilibrium $E_{34}$. In particular, we observe that along the parameter locus $K=K_0=-\frac{4}{b}$ the nonzero eigenvalue $\lambda_2^{34}=\frac{Kb+4}{bK}$ vanishes, so that $E_{34}$ acquires a double zero eigenvalue. Using normal form theory, we determine the codimension and the type of the resulting singularity.

\begin{theorem}
Let $K_0=-\frac{4}{b}$ and $h_0=\frac{-2b^2+4a}{b^3}$, and assume $0<h<-\frac{1}{b}$. For $d=d_0$ and $K=K_0$, the equilibrium $E_{34}$ possesses a zero eigenvalue with algebraic multiplicity two. If $h\ne h_0$, it can be reduced to a generalized cusp normal form of codimension $2$. If, in addition, $2a<b^2<4a$ and $h=h_0$, then $h_0>0$ and the degeneracy increases to codimension $3$.
\end{theorem}

\begin{proof}
In what follows, $\rho_{ij}$ and $\varpi_{ij}$, $i,j\in\mathbb N$, are defined in Appendix~A. Substituting $K=K_0$ into $J(E_{34})$ yields a double zero eigenvalue. We translate $E_{34}$ to the origin via the affine change of variables
\[
x_1=x-x_{lr},\qquad y_1=y-y_{lr},
\]
and compute the Taylor expansion of the resulting vector field in a neighborhood of the origin. Next, applying the local change of variables
\[
x_2=x_1,\qquad y_2=\dot x_1,
\]
system~\eqref{Rescaled Main Model} is converted into
\begin{equation}\label{Cusp step 2}
\begin{cases}
\dot x_2=y_2,\\
\dot y_2=
-\frac{b^3(bh+1)}{(-b^2+4a)^2}x_2^2
+\frac{b(-b^3h-2b^2+4a)}{-2b^2+8a}x_2y_2
+O(|x_2,y_2|^3).
\end{cases}
\end{equation}

For $h\ne h_0$, the normal form~\eqref{Cusp step 2} admits a cusp singularity of codimension~$2$ at the origin, characterized by two hyperbolic sectors separated by two separatrices \cite{Chow_Li_Wang_1994,dumortier1991bifurcations}. The qualitative dynamics of system~\eqref{Rescaled Main Model} are preserved under the coordinate changes applied above; consequently, for the specified parameter values, system~\eqref{Rescaled Main Model} in a neighborhood of $E_{34}$ is locally equivalent to system~\eqref{Cusp step 2} in a neighborhood of the origin. To establish that $E_{34}$ is a cusp singularity of codimension~$3$, we show that our system is $C^\infty$ equivalent to the generalized system introduced in~\cite{dumortier1987generic}:

\begin{equation}\label{cusp 3 standard form}
\begin{cases}
\dot x=y,\\
\dot y=x^2+\vartheta x^3y+O(|x,y|^5).
\end{cases}
\end{equation}
provided $\vartheta$ is a nonzero real coefficient. For $h=h_0$, the coefficient of $x_2y_2$ in (\ref{Cusp step 2}) is equal to zero. Moreover,
\[
h_0>0\quad\Longleftrightarrow\quad b^2>2a,
\]
so the biologically admissible codimension-$3$ case requires $2a<b^2<4a$. Under this condition $h_0<-\frac{1}{b}$, and \eqref{Cusp step 2} becomes of the form:
\begin{equation}\label{eq:cusp-rho-expansion}
\begin{cases}
\dot x_2=y_2,\\
\dot y_2=
\rho_{20}x_2^2+\rho_{30}x_2^3
+\rho_{21}x_2^2y_2+\rho_{12}x_2y_2^2
+\rho_{40}x_2^4+\rho_{31}x_2^3y_2
+\rho_{22}x_2^2y_2^2
+O(|x_2,y_2|^5).
\end{cases}
\end{equation}
The coefficient of $x_2^2$ is normalized by the rescaling $x_3=\rho_{20}x_2,y_3=\rho_{20}y_2$, which in turn gives us:
\begin{equation}\label{eq:cusp-varpi-expansion}
\begin{cases}
\dot x_3=y_3,\\
\dot y_3=x_3^2
+\varpi_{30}x_3^3
+\varpi_{21}x_3^2y_3
+\varpi_{12}x_3y_3^2
+\varpi_{40}x_3^4
+\varpi_{31}x_3^3y_3
+\varpi_{22}x_3^2y_3^2
+O(|x_3,y_3|^5).
\end{cases}
\end{equation}
By proposition (\ref{propostion from Reousou 2008}), $\vartheta=-\frac{(-b^2+4a)^2b^2}{4}\ne0$, and system (\ref{Rescaled Main Model}) is $C^\infty$ equivalent to:
\begin{equation}
\begin{cases}
\dot x_4=y_4,\\
\dot y_4=x_4^2-\frac{(-b^2+4a)^2b^2}{4}x_4^3y_4
+O(|x_4,y_4|^5).
\end{cases}
\end{equation}
\end{proof}

Under the parameter restrictions stated above, $E_{34}$ is a degenerate equilibrium of codimension~$3$. We introduce a three-parameter perturbation and derive its local unfolding to describe the nearby dynamics.

\begin{theorem}\label{BT3 theorem}
Assume $2a<b^2<4a$ and let
\[
\epsilon=(\epsilon_1,\epsilon_2,\epsilon_3)
=(d-d_0,\;K-K_0,\;h-h_0).
\]
Then for $\epsilon\approx0$, system \eqref{Rescaled Main Model} exhibits a codimension 3 Bogdanov--Takens bifurcation near the point $E_{34}$.
\end{theorem}

\begin{proof}
Perturbing $\epsilon$ in a neighborhood of the origin, system~\eqref{Rescaled Main Model} takes the form

\begin{equation}\label{BT system}
\begin{cases}
\dot x=x\left(1-\frac{x}{K_0+\epsilon_2}\right)-p(x)y-h_0-\epsilon_3,\\
\dot y=y(-d_0-\epsilon_1+p(x)).
\end{cases}
\end{equation}

Our objective is to prove that there exists a $C^\infty$ change of variables (a local $C^\infty$ diffeomorphism, depending smoothly on $\epsilon$) that transforms system~\eqref{BT system} into the generalized universal unfolding near a Bogdanov--Takens point:
\begin{equation}
\begin{cases}
\dot x=y,\\
\dot y=\alpha_1+\alpha_2y+\alpha_3xy+x^2-x^3y+\Xi(x,y,\epsilon).
\end{cases}
\end{equation}
provided that $\alpha_1,\alpha_2,$ and $\alpha_3$ are real functions of $\epsilon$,
\[
\Xi(x_i,y_i,\epsilon)
=
O(|x_i,y_i|^5)
+O(\epsilon)
\left(
O(|x_i,y_i|^3)+y_i^2O(|x_i,y_i|)+y_i^2
\right)
+O(\epsilon^2)\left(1+O(|x_i,y_i|)\right),
\]
$i=1,2,3,4,5,6$, and the Jacobian matrix of $(\alpha_1,\alpha_2,\alpha_3)$ with respect to $(\epsilon_1,\epsilon_2,\epsilon_3)$ is nonsingular at $\epsilon=0$, that is,
\[
\left.
\Det\begin{pmatrix}
\dfrac{\partial\alpha_1}{\partial\epsilon_1} &
\dfrac{\partial\alpha_1}{\partial\epsilon_2} &
\dfrac{\partial\alpha_1}{\partial\epsilon_3} \\
\dfrac{\partial\alpha_2}{\partial\epsilon_1} &
\dfrac{\partial\alpha_2}{\partial\epsilon_2} &
\dfrac{\partial\alpha_2}{\partial\epsilon_3} \\
\dfrac{\partial\alpha_3}{\partial\epsilon_1} &
\dfrac{\partial\alpha_3}{\partial\epsilon_2} &
\dfrac{\partial\alpha_3}{\partial\epsilon_3}
\end{pmatrix}
\right|_{\epsilon=0}\ne0.
\]
We follow the reduction procedure of \cite{li2015codimension}. The coefficients for the present system ($m_{ij}$, $n_{ij}$, $p_{ij}$, $q_{ij}$, and $r_{ij}$, $i,j\in\mathbb N$) are listed in Appendix A. We first translate the unperturbed equilibrium $E_{34}$ to the origin. Then, introducing the change of variables
\[
x_1=x,\qquad y_1=\dot x,
\]
system~\eqref{BT system} can be rewritten in the form

\begin{equation}\label{step_1 BT3}
\begin{cases}
\dot x_1=y_1,\\
\dot y_1=
m_{00}+m_{10}x_1+m_{01}y_1
+m_{11}x_1y_1+m_{12}x_1y_1^2
+m_{20}x_1^2+m_{21}x_1^2y_1
+m_{22}x_1^2y_1^2+m_{30}x_1^3\\
\hspace{1cm}
+m_{31}x_1^3y_1+m_{40}x_1^4+\Xi(x_1,y_1,\epsilon).
\end{cases}
\end{equation}

We next aim to remove the nonlinear terms $y_1^2x_1$ and $y_1^2x_1^2$. To this end, we introduce the following near-identity change of variables:
\[
x_1=\phi(x_2),\qquad
y_1=\phi'(x_2)y_2
=\left(1+\frac{m_{12}}{2}x_2^2+\frac{m_{22}}{3}x_2^3\right)y_2,
\]
where
\[
\phi(x_2)=x_2+\frac{m_{12}}{6}x_2^3+\frac{m_{22}}{12}x_2^4.
\]

The system can be rewritten as follows:
\begin{equation}\label{eq:bt-step2}
\begin{cases}
\dot x_2=y_2,\\
\dot y_2=
n_{00}+n_{10}x_2+n_{01}y_2
+n_{11}x_2y_2+n_{20}x_2^2
+n_{21}x_2^2y_2+n_{30}x_2^3
+n_{40}x_2^4+n_{31}x_2^3y_2\\
\hspace{1cm}
+O(\epsilon)O(|x_2,y_2|^5)+\Xi(x_2,y_2,\epsilon).
\end{cases}
\end{equation}
Equivalently, the previous system can be written in differential form as:
\begin{equation}\label{differential form}
y_2\,dy_2=
\left[
\Pi(x_2,y_2)
+n_{20}x_2^2+n_{30}x_2^3+n_{40}x_2^4
+\Xi(x_2,y_2,\epsilon)
\right]dx_2,
\end{equation}
where
\[
\Pi(x_2,y_2)
=
n_{00}+n_{10}x_2
+y_2(n_{01}+n_{11}x_2+n_{21}x_2^2+n_{31}x_2^3).
\]
To drop the $x_2^3$ and $x_2^4$ terms, we solve the differential form:
\[
(n_{20}x_2^2+n_{30}x_2^3+n_{40}x_2^4)\,dx_2
=
n_{20}x_3^2\,dx_3+O(x_3^5)\,dx_3,
\]
and obtain
\[
x_2=\phi(x_3)
=
x_3-\frac{n_{30}}{4n_{20}}x_3^2
+\frac{15n_{30}^2-16n_{20}n_{40}}{80n_{20}^2}x_3^3.
\]
Define $x_3=\phi^{-1}(x_2)$ and $y_3=y_2$. \eqref{differential form} is reformulated as
\[
y_3\,dy_3
=
\left[
\Pi(\phi(x_3),y_3)+\Xi(x_2,y_2,\epsilon)
\right]\phi'(x_3)\,dx_3+n_{20}x_3^2\,dx_3
\]
\[
=
\left[
p_{00}+p_{10}x_3+p_{01}y_3
+p_{11}x_3y_3+p_{20}x_3^2
+p_{21}x_3^2y_3+p_{31}x_3^3y_3
+\Xi(x_3,y_3,\epsilon)
\right]dx_3,
\]
where $m_{10}$ and $n_{10}$ are $O(\epsilon^2)$ and $p_{10}$ is $O(\epsilon)$, and $p_{30}x_3^3$ and $p_{40}x_3^4$ are of $O(\epsilon)$. The transformation can be summarized with the near identity transformation:
\[
x_2=x_3-\frac{n_{30}}{4n_{20}}x_3^2
+\frac{15n_{30}^2-16n_{20}n_{40}}{80n_{20}^2}x_3^3,
\qquad
y_2=y_3,\qquad
\frac{dt}{d\tau}=\phi'(x_3)>0.
\]
Replace $\tau$ by $t$, and the system has the form:
\begin{equation}\label{ste$P_3$ BT3}
\begin{cases}
\dot x_3=y_3,\\
\dot y_3=
p_{00}+p_{10}x_3+p_{01}y_3
+p_{11}x_3y_3+p_{20}x_3^2
+p_{21}x_3^2y_3+p_{31}x_3^3y_3\\
\hspace{1cm}+\Xi(x_3,y_3,\epsilon).
\end{cases}
\end{equation}
To eliminate the $x_3^2y_3$ term, we employ the near identity transformation:
\[
x_3=x_4,\qquad
y_3=y_4+\frac{p_{21}}{3p_{20}}y_4^2
+\frac{p_{21}^2}{36p_{20}^2}y_4^3.
\]
Subsequently, we rescale time once again by setting
\[
\frac{dt}{d\tau}
=
\left(
1+\frac{p_{21}}{3p_{20}}y_4
+\frac{p_{21}^2}{36p_{20}^2}y_4^2
\right)^{-1}
\]
and redenoting $\tau$ by $t$. System (\ref{ste$P_3$ BT3}) becomes:
\begin{equation}\label{BT step 4}
\begin{cases}
\dot x_4=y_4,\\
\dot y_4=
q_{00}+q_{01}y_4+q_{10}x_4+q_{11}x_4y_4
+q_{20}x_4^2+q_{31}x_4^3y_4+\Xi(x_4,y_4,\epsilon).
\end{cases}
\end{equation}
The coefficients of $x_4^2$ and $x_4^3y_4$ are changed to $1$ and $-1$, respectively, by applying the linear change of variables:
\[
x_4=q_{20}^{\frac{1}{5}}(-q_{31})^{-\frac{2}{5}}x_5,\qquad
y_4=q_{20}^{\frac{4}{5}}(-q_{31})^{-\frac{3}{5}}y_5,\qquad
t=q_{20}^{-\frac{3}{5}}(-q_{31})^{\frac{1}{5}}\tau,
\]
and equations \eqref{BT step 4} take the form:
\begin{equation}\label{eq:bt-step5}
\begin{cases}
\dot x_5=y_5,\\
\dot y_5=
r_{00}+r_{01}y_5+r_{10}x_5+r_{11}x_5y_5
+x_5^2-x_5^3y_5+\Xi(x_5,y_5,\epsilon).
\end{cases}
\end{equation}
Finally, applying the near-identity transformation:
$x_5=x_6-\frac{r_{10}}{2}$ and $y_5=y_6$, we discard the $x_5$ term, and system (\ref{Rescaled Main Model}) localized at $E_{34}$ is locally $C^\infty$-orbitally equivalent to
\begin{equation}
\begin{cases}
\dot x_6=y_6,\\
\dot y_6=
\alpha_1+\alpha_2y_6+\alpha_3x_6y_6
+x_6^2-x_6^3y_6+\Xi(x_6,y_6,\epsilon),
\end{cases}
\end{equation}
provided that
\[
\alpha_1
=
\frac{(-b^2+4a)^{\frac{8}{5}}}{2^{\frac{8}{5}}(-b)^{\frac{2}{5}}}\epsilon_1
+O(\epsilon^2),
\]
\[
\alpha_2
=
\frac{2^{\frac{3}{5}}(-b^2+4a)^{\frac{2}{5}}}{8(-b)^{\frac{8}{5}}}
\left(3(-b^2+4a)^2\epsilon_1-b^3\epsilon_2\right)
+O(\epsilon^2),
\]
and
\[
\alpha_3
=
\frac{2^{\frac{2}{5}}}{8(-b)^{\frac{12}{5}}(-b^2+4a)^{\frac{2}{5}}}
\left(
3(-b^2+4a)^3\epsilon_1
-b^3(8a-b^2)\epsilon_2+4b^5\epsilon_3
\right)
+O(\epsilon^2).
\]
Finally, we compute
\[
\Det\!\left(\frac{\partial(\alpha_1,\alpha_2,\alpha_3)}
{\partial(\epsilon_1,\epsilon_2,\epsilon_3)}\right)
=
-2^{-\frac{23}{5}}(-b^2+4a)^{\frac{8}{5}}(-b)^{\frac{18}{5}}
+O(\epsilon),
\]
which is nonzero for sufficiently small $\epsilon$. This finishes the proof.
\end{proof}

The local bifurcation structure is illustrated schematically in Fig. \ref{fig:placeholder}; see \cite{dumortier1987generic,lamontagne2008bifurcation}. Here, $LC$ denotes a limit cycle, while $BT_2^+$ and $BT_2^-$ denote codimension-two Bogdanov--Takens points on two branches of the unfolding. The labels $SN^\pm$, $H^\pm$, and $HL^\pm$ indicate saddle-node, Hopf, and homoclinic-loop branches, respectively, and $DC$ indicates a double-cycle bifurcation. The superscripts distinguish the branches shown in the figure. In this Bogdanov--Takens diagram, $HL$ refers to a homoclinic loop. The point $GH$ is a generalized Hopf point, $HL_2$ is a codimension-two homoclinic-loop point, and $C$ marks a simultaneous Hopf and homoclinic-loop bifurcation.

\begin{figure}[!htbp]
    \centering
    \includegraphics[width=0.8\linewidth]{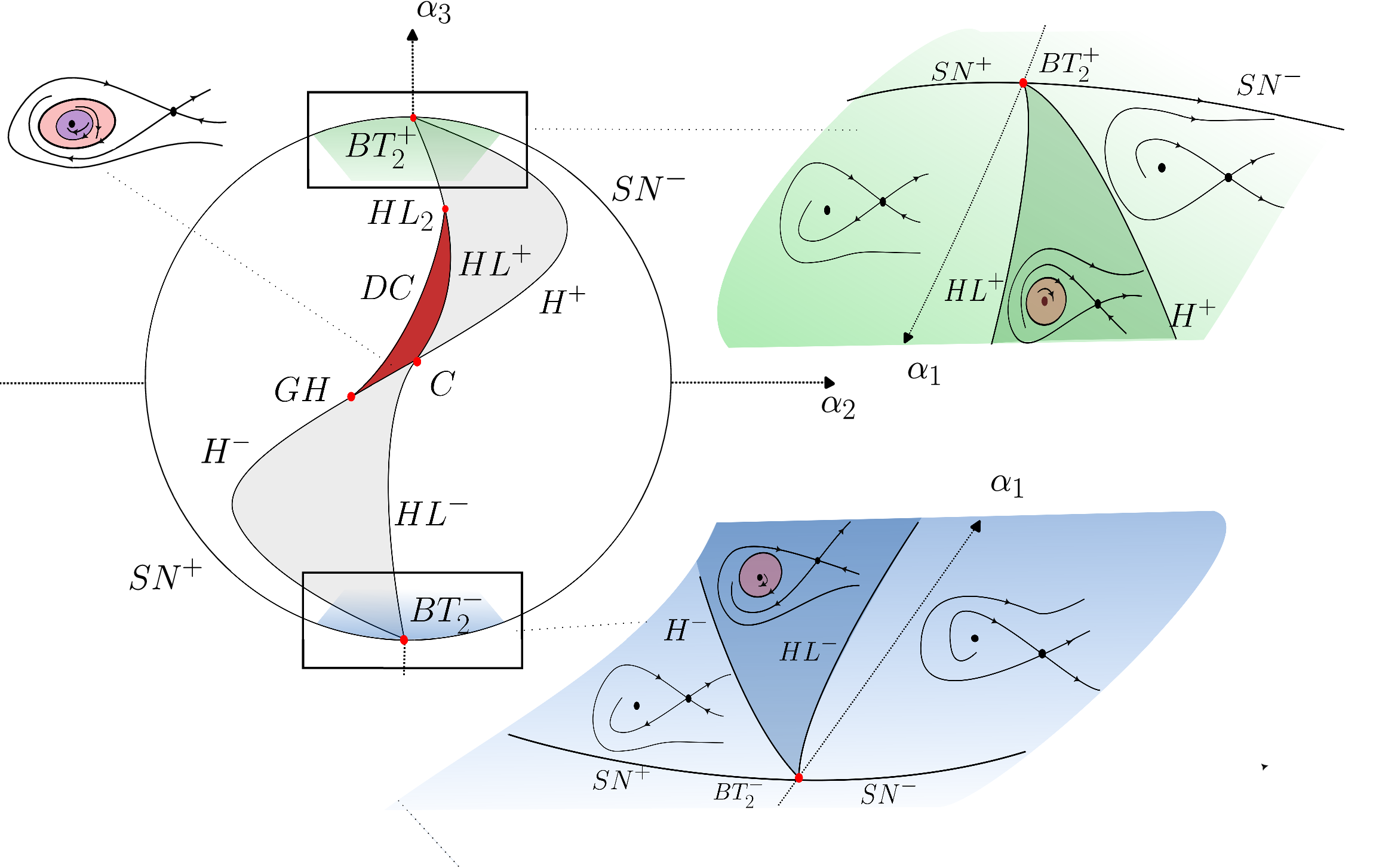}
    \caption{Bifurcation diagram near a cusp singularity of codimension three.
The curves \(SN^\pm\), \(H^\pm\), \(HL^\pm\), and \(DC\) represent
saddle-node, Hopf, homoclinic loop, and double-cycle bifurcations,
respectively. The labels \(BT_2^\pm\), \(GH\), \(HL_2\), and \(C\) denote codimension-two Bogdanov--Takens points, a generalized Hopf point, a codimension-two homoclinic-loop point, and a simultaneous Hopf and homoclinic-loop bifurcation, respectively.}
    \label{fig:placeholder}
\end{figure}

The Bogdanov--Takens unfolding helps us understand how Hopf bifurcations arise near the degenerate equilibrium. As the parameters vary, the two zero eigenvalues may become a pair of complex conjugate eigenvalues. A Hopf bifurcation occurs when this pair crosses the imaginary axis. We can therefore use the unfolding to locate Hopf bifurcations and study their dependence on the model parameters.

We are particularly interested in degenerate Hopf bifurcations, which occur when the first Lyapunov coefficient vanishes. Such points may give rise to multiple limit cycles. In the following section, we investigate these bifurcations near the coexistence equilibrium $E_3$ and determine the conditions under which one or two limit cycles can occur.

\begin{figure}[!htbp]
  \centering
  \begin{subfigure}{0.3\linewidth}
    \centering
    \includegraphics[width=\linewidth]{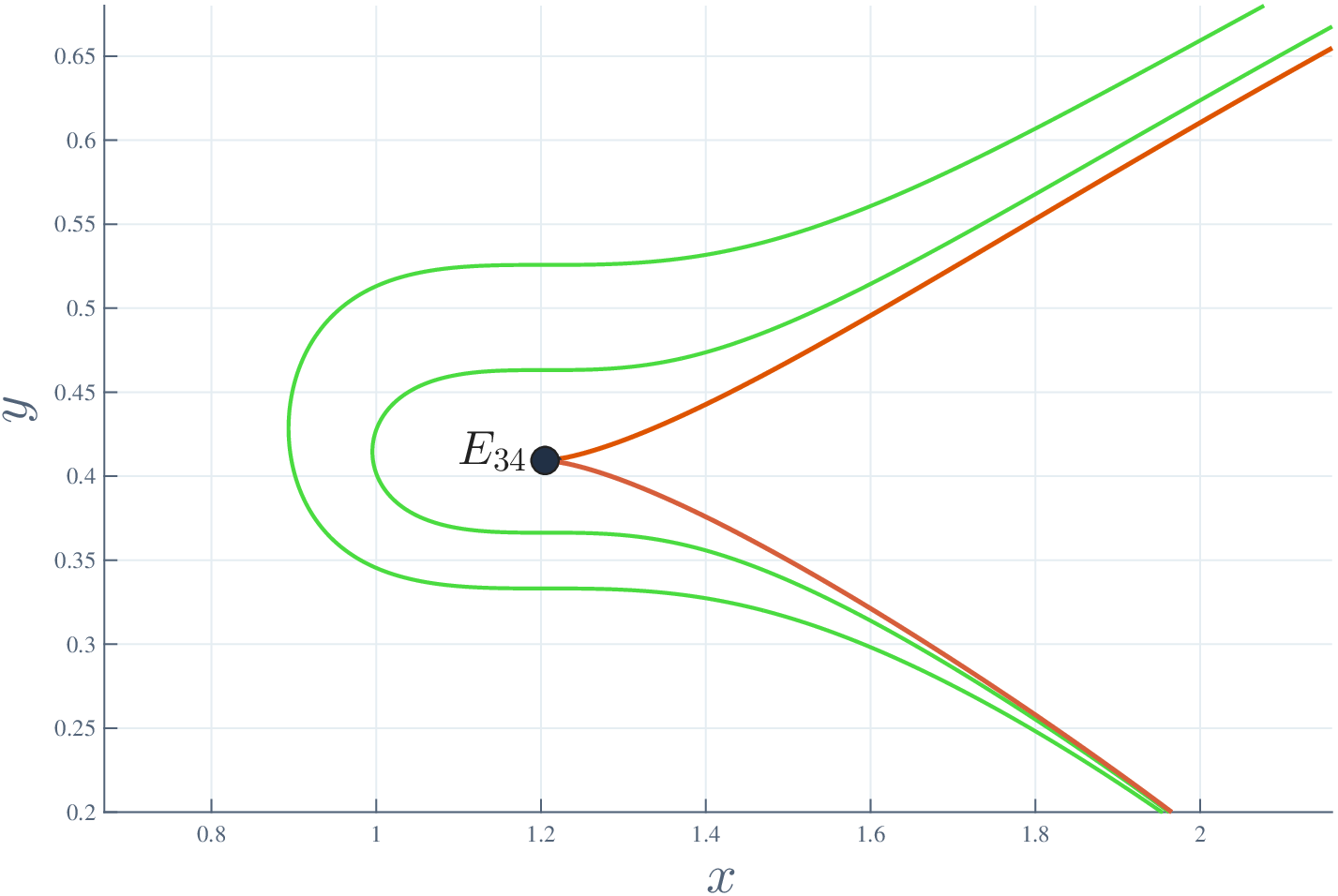}
    \caption{Cusp point}
    \label{fig:Cusp point}
  \end{subfigure}
  \hspace{0.1cm}
  \begin{subfigure}{0.3\linewidth}
    \centering
    \includegraphics[width=\linewidth]{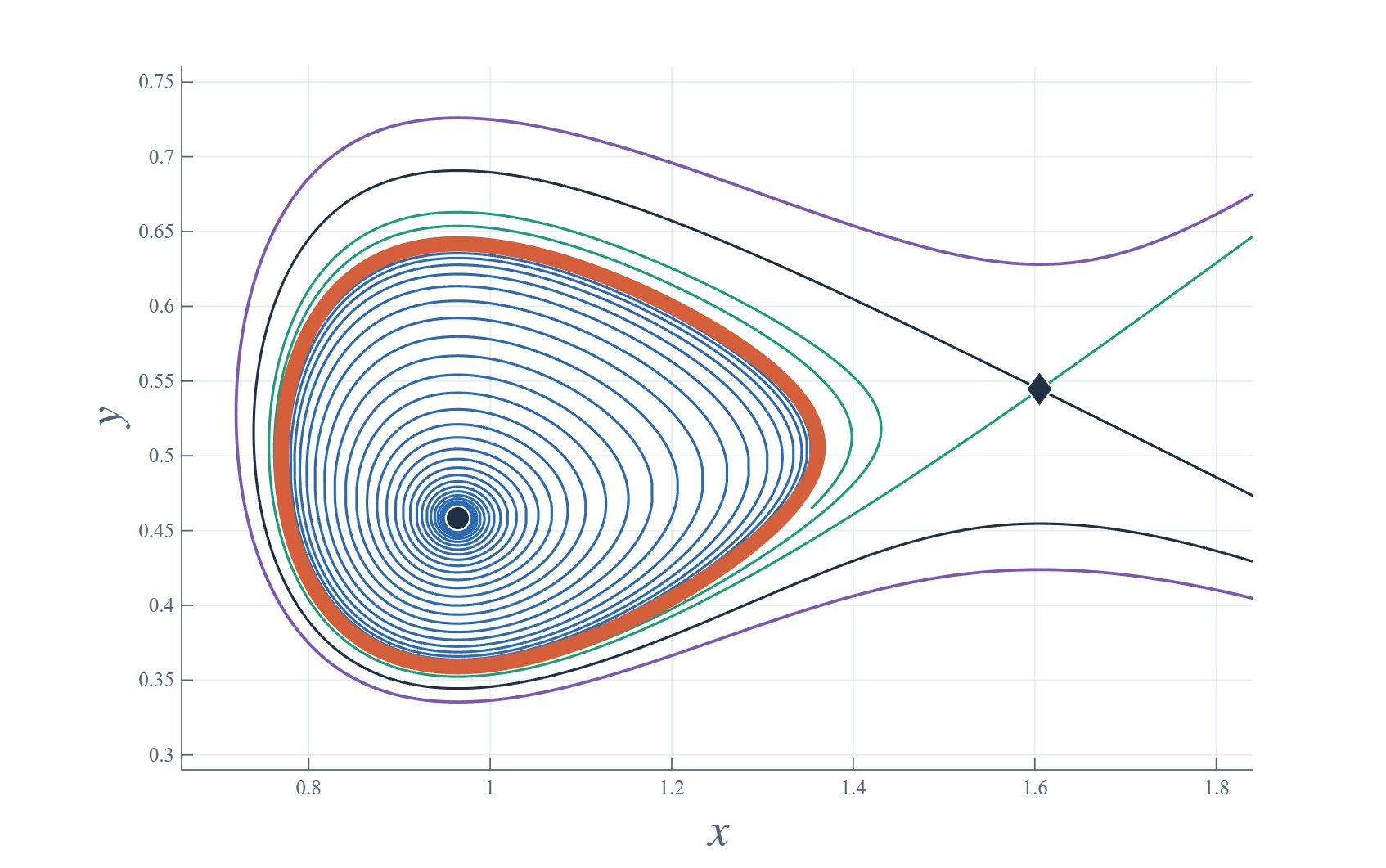}
    \caption{Limit cycle near cusp point}
    \label{fig:limit cycle near cusp point}
  \end{subfigure}
  \hspace{0.1cm}
  \begin{subfigure}{0.3\linewidth}
    \centering
    \includegraphics[width=\linewidth]{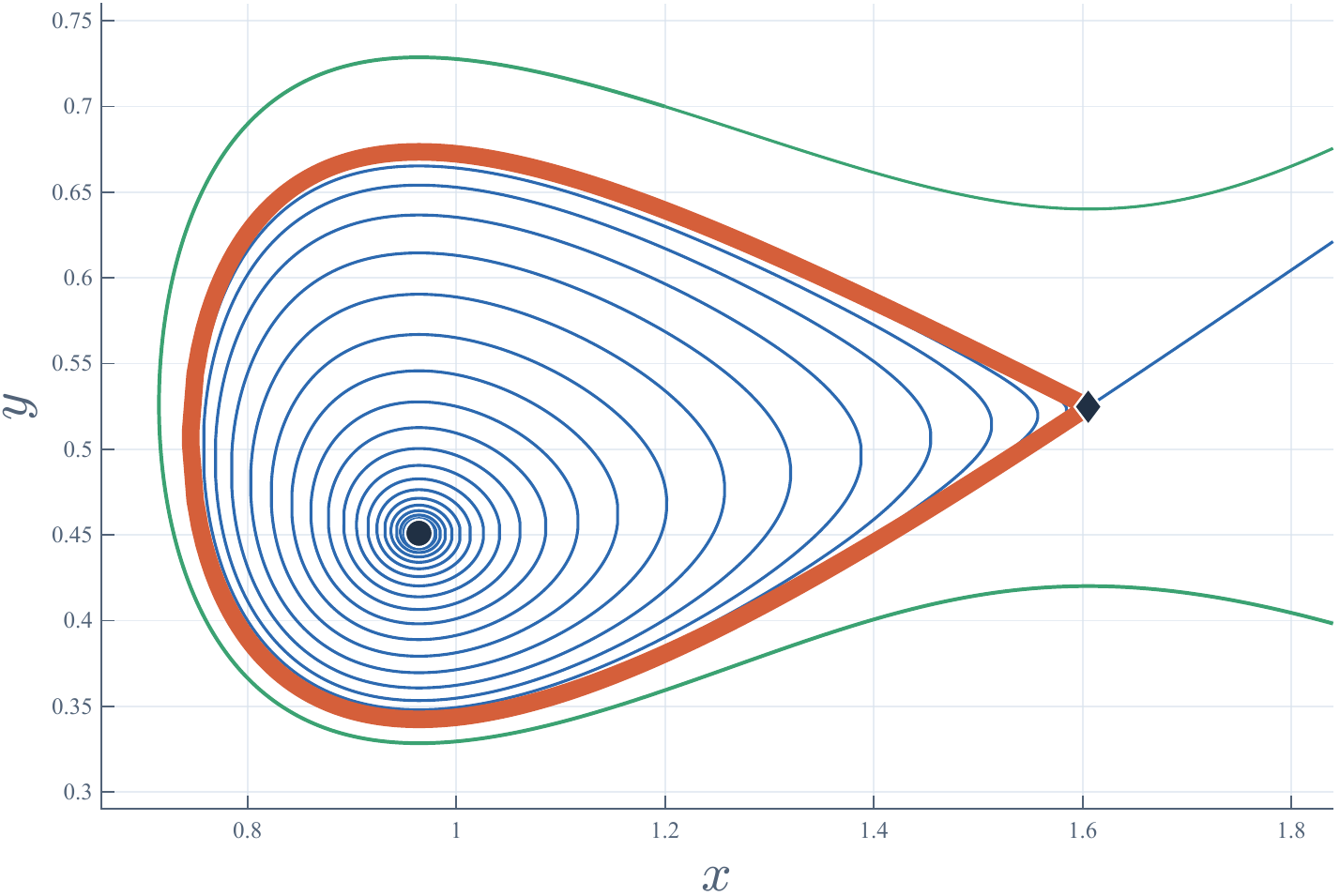}
    \caption{Homoclinic loop}
    \label{fig:Homoclinic loop}
  \end{subfigure}
  \caption{The graphs above illustrate local dynamics near the Bogdanov-Takens bifurcation point with $a=1.373$, $b=-1.660$, $d=d_0=\frac{4}{4a-b^2}$, $K=K_0=-\frac{4}{b}$, and $h=h_0=\frac{4a-2b^2}{b^3}$. At the cusp, $E_3$ and $E_4$ coincide. In the perturbed panels with two interior equilibria, the rightmost point is $E_4$ (saddle type), and the leftmost point is $E_3$ (focus).}
  \label{fig:bt_local_phase_portraits}
\end{figure}

\FloatBarrier
\subsection{Codimension 2 Hopf Bifurcation}
We use the criterion in \cite{xiao2006multiple,lyapunov} to determine the weak-focus multiplicity of $E_3$ and the parameter conditions for one or two nearby limit cycles.

\begin{lemma}\label{Lienard System}
Consider the Lienard system:
\begin{equation}
\begin{cases}
\dot x=\Lambda_1(y)-\Lambda_2(x),\\
\dot y=-\Lambda_3(x),
\end{cases}
\end{equation}
where $\Lambda_1(y),\Lambda_2(x)$, and $\Lambda_3(x)$ are $C^\infty$ smooth functions of their variables in a neighborhood of the origin, and $\Lambda_1(0)=\Lambda_2(0)=\Lambda_3(0)=\Lambda_2'(0)=0$, $\Lambda_1'(0)>0$, $x\Lambda_3(x)>0$ for $x\ne0$, and $\Lambda_3'(0)>0$. Let $I(x)=\int_0^x\Lambda_3(s)\,ds$. If there exists a $C^\infty$ smooth function $\sigma(x)=-x+O(x^2)$ such that $I(\sigma(x))\equiv I(x)$ and
\[
\Lambda_2(\sigma(x))-\Lambda_2(x)
=
\sum_{i=1}^{2k+1}B_i x^i+O(x^{2k+2}),
\qquad k\ge1,
\]
then the equilibrium point $(0,0)$ of the generalized Lienard system is a weak focus of multiplicity $k$ if $B_j=0$, $j=1,2,\ldots,2k$, and $B_{2k+1}\ne0$. It is locally stable (respectively unstable) if $B_{2k+1}<0$ ($B_{2k+1}>0$).
\end{lemma}

\begin{theorem}
For system~\eqref{Rescaled Main Model}, the multiplicity of the weak focus $E_3$ is at most two in the biologically admissible region $h>0$. Consequently, at most two small-amplitude limit cycles can bifurcate locally from $E_3$. If
\[
B_1=B_2=B_3=B_4=0,\qquad B_5\ne0,
\]
and
\[
\operatorname{rank}
\frac{\partial(B_1,B_3)}{\partial(K,h)}=2,
\]
then $E_3$ is a codimension-$2$ generalized Hopf (Bautin) point.
\end{theorem}

\begin{proof}
We begin by performing a time reparametrization of system~\eqref{Rescaled Main Model}. Define
\[
\tau=\int_0^t p(x(s))\,ds,
\]
so that $\frac{d\tau}{dt}=p(x(t))$ and, whenever $p(x)>0$, differentiation with respect to $\tau$ is well defined. Writing derivatives in the new time variable yields the equivalent system
\begin{equation}\label{eq:tau_system}
\begin{cases}
\frac{dx}{d\tau}
=
\frac{x(1-\frac{x}{K})-h}{p(x)}-y,\\[0.8em]
\frac{dy}{d\tau}
=
y\left(1-\frac{d}{p(x)}\right).
\end{cases}
\end{equation}
Next, we translate the phase variables about a reference point $(x_\ell,y_\ell)$ by introducing
\[
x=x_\ell-u,\qquad y=\mathcal T(v)+y_\ell,
\]
where $\mathcal T$ is chosen so that the $y$-translation is compatible with an exponential reparametrization. Specifically, let $\mathcal T(v)$ solve the initial value problem
\[
\frac{d\mathcal T(v)}{dv}=\mathcal T(v)+y_\ell,\qquad \mathcal T(0)=0.
\]
This linear ODE has the explicit solution $\mathcal T(v)=y_\ell(e^v-1)$, and hence $y=y_\ell e^v$. With these transformations, system~\eqref{eq:tau_system} admits the normal form representation
\begin{equation}\label{eq:uv_system}
\begin{cases}
\dot u=\Lambda_1(v)-\Lambda_2(u),\\
\dot v=-\Lambda_3(u),
\end{cases}
\end{equation}
where the nonlinearities are given explicitly by
\[
\Lambda_1(v)=y_\ell(e^v-1),
\qquad
\Lambda_2(u)
=
\left.
\frac{(x(1-\frac{x}{K})-h)(ax^2+bx+1)}{x^2}
\right|_{x=x_\ell-u}
-y_\ell,
\]
and
\[
\Lambda_3(u)
=
-\frac{u\left(-bx_\ell^2+(bu-2)x_\ell+u\right)}
{(ax_\ell^2+bx_\ell+1)(u-x_\ell)^2}.
\]
To facilitate subsequent analysis, define the primitive
\[
I(u)=\int_0^u\Lambda_3(s)\,ds.
\]
We seek a local function $\sigma(u)$ satisfying
\[
\sigma(u)=-u+O(u^2)
\]
and determined implicitly by the functional identity
\[
I(\sigma(u))-I(u)\equiv0.
\]
The coefficients of $\sigma(u)$ of order at least two are obtained by expanding this identity in powers of $u$ and matching coefficients. In particular, using Maple we compute the formal series representation
\[
\sigma(u)=-u+\sum_{i=2}^{\infty}\gamma_i u^i,
\]
where the coefficients $\gamma_i\in\mathbb R$ are uniquely determined by the condition $I(\sigma(u))\equiv I(u)$ in a neighborhood of $u=0$ and are given by
\begin{align*}
\gamma_2&=-\frac{2(bx_\ell+3)}{3x_\ell(bx_\ell+2)},&
\gamma_3&=-\frac{4(b^2x_\ell^2+6bx_\ell+9)}{9x_\ell^2(bx_\ell+2)^2},\\
\gamma_4&=-\frac{4(11b^3x_\ell^3+99b^2x_\ell^2+279bx_\ell+270)}
{135x_\ell^3(bx_\ell+2)^3},&
\gamma_5&=-\frac{8(13b^4x_\ell^4+156b^3x_\ell^3+648b^2x_\ell^2+1161bx_\ell+810)}
{405x_\ell^4(bx_\ell+2)^4}.
\end{align*}
Building on Lemma \ref{Lienard System}, we compute $B_1$ and $B_2$ as follows:
\begin{align}
B_1&=
\frac{-4ax_\ell^4+(2aK-2b)x_\ell^3+(2bh-2)Kx_\ell+4hK}
{Kx_\ell^3},\qquad
B_2=\frac{(bx_\ell+3)B_1}{3x_\ell(bx_\ell+2)}.
\end{align}
These coefficients vanish simultaneously when
\[
h=h_H=
-\frac{x_\ell\left(-2ax_\ell^3+(aK-b)x_\ell^2-K\right)}
{K(bx_\ell+2)}.
\]
On the locus $B_1=0$, we obtain the next pair of coefficients $B_3$ and $B_4$ as
\begin{align}
B_3&=
\frac{
-2b(abK-b^2-2a)x_\ell^4
-12b(aK-b)x_\ell^3
+2(-6aK+6b)x_\ell^2
+8bKx_\ell+12K}
{3K(bx_\ell+2)^2x_\ell^4},\qquad
B_4=\frac{(bx_\ell+3)B_3}{x_\ell(bx_\ell+2)}.
\end{align}
These vanish precisely when
\[
K=K_{\mathrm{GH}}=
\frac{bx_\ell^2(b^2x_\ell^2+2ax_\ell^2+6bx_\ell+6)}
{ab^2x_\ell^4+6abx_\ell^3+6ax_\ell^2-4bx_\ell-6}.
\]
After imposing $B_1=B_3=0$, we find that the higher-order coefficients $B_5$ and $B_6$ vanish at $(a,b,d)=(a^\ast,b^\ast,d^\ast)$ only if the following constraint holds:
\begin{equation}\label{eq:F_constraint}
\mathcal F(a,b,d)
:=
a^2b^2x_\ell^5+11a^2bx_\ell^4+18a^2x_\ell^3
-4abx_\ell^2+(4b^2-18a)x_\ell+3b
=0.
\end{equation}

Equivalently $a$ satisfies
\[
a=a_{H_3}
=
\frac{
2bx_\ell\pm\sqrt{
-4b^4x_\ell^4-47b^3x_\ell^3-101b^2x_\ell^2-18bx_\ell+81
}+9}
{x_\ell^2(b^2x_\ell^2+11bx_\ell+18)}.
\]

Substituting $a_{H_3}$ into $K_{\mathrm{GH}}$ and then into $h_H$ gives the candidate harvesting value
\[
h_{H_3}=\frac{(bx_\ell-3)x_\ell}{(bx_\ell+9)(bx_\ell+2)}.
\]

Since $b<0$ and $x_\ell<-\frac{2}{b}$, $h_{H_3}<0$, and hence $B_3=0$ and $B_5=0$ cannot hold simultaneously on the weak-focus locus $B_1=0$ in the biologically admissible region $h>0$. Therefore a weak focus in this region has multiplicity at most two.

To verify the codimension-$2$ nondegeneracy condition, we take $K$ and $h$ as the unfolding parameters, holding $a$, $b$, and $d$ fixed. Since $p(x_\ell)=d$, the value of $x_\ell$ is independent of $K$ and $h$. Direct calculation at $B_1=B_3=0$ gives
\[
\operatorname{rank}
\frac{\partial(B_1,B_3)}{\partial(K,h)}=2,
\]
because
\[
\Det\!\left(\frac{\partial(B_1,B_3)}{\partial(K,h)}\right)
=
\frac{4b(b^2x_\ell^2+2ax_\ell^2+6bx_\ell+6)}
{3K^2x_\ell^5(bx_\ell+2)}\ne0.
\]
Indeed, $b<0$, $x_\ell>0$, and $x_\ell<-\frac{2}{b}$ imply $bx_\ell+2>0$. Moreover, $b^2<4a$ yields
\[
b^2x_\ell^2+2ax_\ell^2+6bx_\ell+6
>
\frac{3}{2}(bx_\ell+2)^2>0,
\]
which establishes the stated rank condition.
Therefore, whenever
\[
B_1=B_2=B_3=B_4=0,\qquad B_5\ne0,
\]
the equilibrium $E_3$ is a codimension-$2$ generalized Hopf (Bautin) point. This finishes the proof.
\end{proof}

Having established that the simultaneous higher-order degeneracy conditions defining \eqref{eq:F_constraint} admit no biologically admissible solution with $h>0$, and hence that at most two small-amplitude limit cycles can bifurcate locally from $E_3$ (Fig. \ref{fig:double limit cycle}), we now shift our attention to global bifurcation phenomena associated with the coalescence of the predator-free equilibria with either $E_3$ or $E_4$. In particular, we study the non-hyperbolic structures that arise at the merging parameter values and determine the local type of the resulting degenerate equilibrium by examining its Jacobian, higher-order terms, and the corresponding normal-form coefficients. This analysis provides the necessary local description of the phase portrait near the coalesced equilibrium and supplies the unfolding framework needed to detect the global objects generated by the degeneracy.
\begin{figure}[!htbp]
    \centering
    \includegraphics[width=0.6\linewidth]{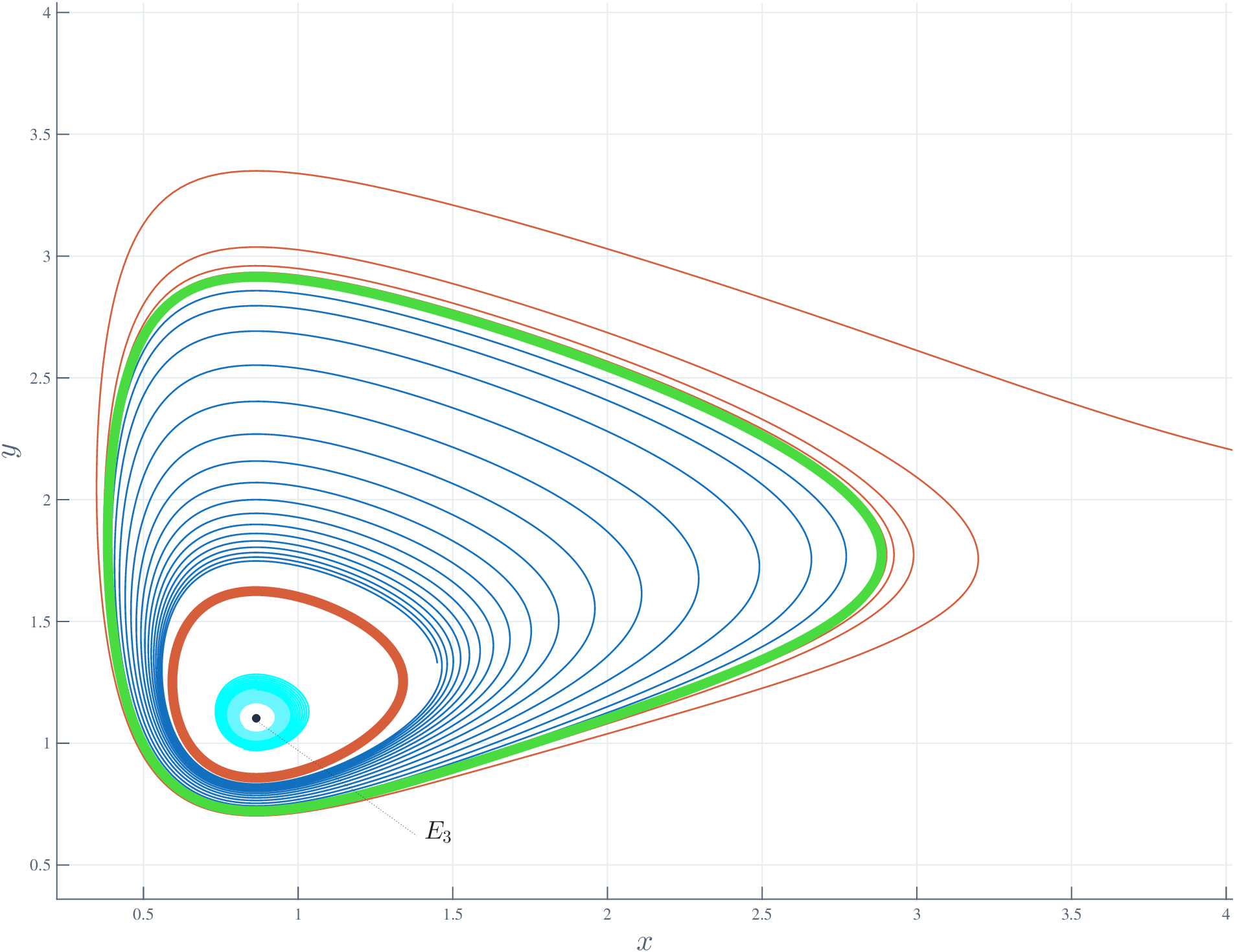}
    \caption{Two nested limit cycles at $a=1.6;\ b=-1.23;\ d=0.66;\ h=0.001;\ K=5.5;$}
    \label{fig:double limit cycle}
\end{figure}

\FloatBarrier
\subsection{Codimension 3 Nilpotent Saddle Bifurcation}

\begin{theorem}\label{Nilpotent Saddle}
Assume $bh>-\frac{3}{2}$, and $bh\ne-1$. For $K=K_2=4h$ and $d=p(2h)=\frac{4h^2}{4ah^2+2bh+1}$, the predator-free equilibria $E_1$ and $E_2$ coalesce with exactly one of the interior equilibria $E_3$ or $E_4$, yielding a triple non-hyperbolic equilibrium. In local coordinates centered at the coalesced equilibrium, the vector field admits the following normal-form expansion:
\begin{equation}\label{Target form NS}
\begin{cases}
\dot x=y+\kappa_1x^2,\\
\dot y=y(x+\kappa_2x^2+\kappa_3x^3+\kappa_4x^4+O(x^5)).
\end{cases}
\end{equation}
where $\kappa_i\in\mathbb R$ are defined below. For $-1<bh<0$, $E_3$ merges with $E_1$ and $E_2$ to form a nilpotent saddle singularity. For $-\frac{3}{2}<bh<-1$, $E_4$ merges with $E_1$ and $E_2$ to form a nilpotent elliptic singularity. Moreover, for
\[
a=a_2=\frac{2b^2h^2+4bh+3}{4h^2(2bh+3)},
\]
the coefficient $\kappa_2$ vanishes while $\kappa_3\ne0$ and $\kappa_4\ne0$; in the nilpotent-saddle regime $-1<bh<0$, the resulting nilpotent saddle is of codimension~$3$.
\end{theorem}

\begin{proof}
Following the invariant-line normal-form procedure in \cite{etoua2010bifurcation}, we give the reduction and coefficients for the present parameter regime below. The reduction of \eqref{Rescaled Main Model} to the target form \eqref{Target form NS} in a neighborhood of the coalescence point $E_{123}$ or $E_{124}$ begins by normalizing the coefficient of $y$ in the first equation to unity. To this end, we perform the time reparametrization
\[
\tau=\int_0^t p(x(s))\,ds
\]
followed by the change of coordinates:
\[
x=x_1+2h,\qquad y=-y_1,
\]
which moves the coalesced equilibrium to the origin and fixes the orientation of the phase portrait. Next, to eliminate higher-order terms in $\dot x_1$ and obtain the canonical nilpotent structure in the first component, we apply a near-identity change of variables generated by a smooth function
\begin{equation}\label{H definition}
H(x_1)
=
x_1\sqrt{
\frac{\dot x_1-y_1}
{-\frac{4ah^2+2bh+1}{16h^3}x_1^2}
},
\qquad H(0)=0,
\end{equation}
where the quotient is extended smoothly to $x_1=0$ and the smooth square-root branch with value $1$ at the origin is chosen. The transformation is then applied as follows:
\[
x_2=H(x_1),\qquad y_2=y_1,\qquad
\frac{d\tau_2}{d\tau}=H'(x_1),
\]
where $\tau_2$ is subsequently relabeled as the time variable. Finally, we apply the rescaling
\[
x_2=\frac{h(4ah^2+2bh+1)}{bh+1}x_3,\qquad
y_2=\frac{h(4ah^2+2bh+1)}{bh+1}y_3
\]
to normalize the coefficient of $x_2y_2$ in the second equation and obtain
\[
\kappa_1=-\frac{(4ah^2+2bh+1)^2}{16h^2(bh+1)},
\qquad
\kappa_2=-\frac{8abh^3-2b^2h^2+12ah^2-4bh-3}{4(bh+1)^2},
\]
\[
\kappa_3
=
\frac{
16a^2bh^5-16ab^2h^4+32a^2h^4
-40abh^3-32ah^2+bh+2}
{4(bh+1)^3},
\qquad
\kappa_4=\kappa_4(a,b,h).
\]
The coefficient of $x_3^2y_3$ is equal to $0$ for $a=a_2$, which is well defined and positive for $bh>-\frac{3}{2}$. Accordingly, for $a=a_2$, we
\[
\bar\kappa_1=-\frac{9(bh+1)^3}{4(2bh+3)^2h^2},
\qquad
\bar\kappa_3=-\frac{3(b^2h^2+3bh+3)}{(2bh+3)^2},
\qquad
\bar\kappa_4=
\frac{3(2b^2h^2+9bh+12)bh}{2(2bh+3)^3}.
\]
Since
\[
b^2h^2+3bh+3=\left(bh+\frac{3}{2}\right)^2+\frac{3}{4}>0,
\]
we have $\bar\kappa_3\ne0$ for every real admissible $bh\ne-\frac{3}{2}$. Moreover, $2b^2h^2+9bh+12>0$ for every real $bh$, so $\bar\kappa_4\ne0$ in the present negative-$b$ regime. This is the nonzero even-order coefficient used in the codimension-$3$ nilpotent-saddle classification. A direct computation shows that the identity of the interior equilibrium participating in the triple coalescence depends on the sign of $bh+1$: if $-1<bh<0$, then $E_3$ coalesces with $E_{12}$, whereas if $-\frac{3}{2}<bh<-1$, then $E_4$ coalesces with $E_{12}$.
Moreover, the qualitative type of the resulting non-hyperbolic equilibrium is determined by the sign of $\kappa_1$ (equivalently, $\operatorname{sign}(\kappa_1)=\operatorname{sign}(-bh-1)$). In particular, $\kappa_1<0$ corresponds to a nilpotent saddle, while $\kappa_1>0$ yields a nilpotent elliptic singularity.
In particular, $\bar\kappa_1<0$ corresponds to a nilpotent saddle, while $\bar\kappa_1>0$ yields a nilpotent elliptic singularity.
\end{proof}

For the codimension-three parameter value, write $d_2:=\left.p(2h)\right|_{a=a_2}=\frac{2h^2(2bh+3)}{3(bh+1)^2}$. Guided by the triple coalescence established above, we now turn to the unfolding of the resulting nilpotent singularity. Although the collision can produce either a nilpotent saddle or a nilpotent elliptic point, we focus on the nilpotent saddle case with an \textbf{invariant axis} through the collision point.\\

To study the nearby bifurcations, we vary \((K,d,a)\) in a small neighborhood of \((K_2,d_2,a_2)\). For \(h<-\frac{1}{b}\), the interior equilibrium involved in the coalescence is \(E_3\). As the unfolded vector field is perturbed in a neighborhood of $(K_2,d_2,a_2)$, the separatrix structure associated with the nearby boundary saddles $E_1$ and $E_2$ may close into a heteroclinic loop enclosing $E_3$, while simultaneously allowing the creation and interaction of limit cycles and saddle-node bifurcations. The next result formalizes this mechanism.

\begin{theorem}
Assume $-1<bh<0$ and
\begin{equation}\label{NS nondegeneracy condition}
9(bh+1)^3-4h^2(2bh+3)^2\ne0.
\end{equation}
Then system~\eqref{Rescaled Main Model}, in a sufficiently small neighborhood of $(K_2,d_2,a_2)$, admits a three-parameter versal unfolding, within the class preserving the invariant $x$-axis, of the codimension-$3$ nilpotent saddle at $E_{123}$.
\end{theorem}

\begin{proof}
In this proof, we show that, after smooth changes of variables, orbital time reparametrizations, and parameters, system~\eqref{Rescaled Main Model} is locally $C^\infty$-orbitally equivalent to the canonical nilpotent-saddle unfolding normal form. Consequently, the bifurcation scenarios appearing in the versal unfolding \cite{etoua2010bifurcation} persist in the original system in a neighborhood of $(K_2,d_2,a_2)$. For system~\eqref{Rescaled Main Model} at parameters $(K_2,d_2,a_2)$ under the assumptions of the theorem, we compute a three-parameter versal unfolding, which clarifies the bifurcation structure near the degenerate point $E_{123}$, which in turn acts as a local organizing center for nearby codimension 1 and codimension-2 transitions in the local phase portrait. To introduce perturbations across the bifurcation thresholds, we set
\[
K=4h+\mu_1,\qquad
d=\frac{4bh^3+6h^2}{3(bh+1)^2}+\mu_2,\quad
\text{and}\quad
a=\frac{2b^2h^2+4bh+3}{4h^2(2bh+3)}+\mu_3,
\]
where $\mu=(\mu_1,\mu_2,\mu_3)$ is the nilpotent-saddle perturbation vector and its components are sufficiently small. With this parametrization, $h$ and $b$ remain as the only free variables. The coefficients $\ell_{ij}$, $c_{ij}$, $d_{ij}$, $e_{ij}$, $f_{ij}$, $g_{ij}$, $h_{ij}$, and $b_{ij}$ are listed in Appendix~A. Following the procedure in the proof of Theorem~\ref{Nilpotent Saddle}, we translate the equilibrium $(2h,0)$ to the origin and apply the time rescaling
\[
\tau=\int_0^t p(x(s))\,ds,
\]
and subsequently replace $\tau$ by $t$. The system then takes the form:
\begin{equation}\label{eq:ns-step1}
\begin{cases}
\dot x_1=
\ell_{00}+\ell_{10}x_1+y_1+\ell_{20}x_1^2
+\ell_{30}x_1^3+\ell_{40}x_1^4+O(x_1^5),\\
\dot y_1=
c_{01}y_1+c_{11}x_1y_1+c_{21}x_1^2y_1
+c_{31}x_1^3y_1+y_1O(x_1^4).
\end{cases}
\end{equation}
where the coefficients of $x_1^j$, $j=3,4,5,\ldots$, in the first equation may have nonzero $\mu^0$ parts. To get rid of all constant $(\mu^0)$ coefficients of $x^j$, $j=3,4,5,\ldots$, the following transformation is then applied:
\[
x_2=H(x_1),\qquad y_2=y_1,\qquad
\frac{d\tau}{dt}
=
H'(x_1)
=
\frac{1}{\frac{d}{dx_2}(H^{-1}(x_2))},
\]
provided that $H$ is defined as in \eqref{H definition} at the unperturbed values $K=K_2$, $a=a_2$, and $\mu=0$, and is kept fixed as $\mu$ varies. Equivalently,
\[
H(x_1)
=
x_1\frac{2h}{2h+x_1}
\sqrt{
\frac{a_2(2h+x_1)^2+b(2h+x_1)+1}
{4a_2h^2+2bh+1}
}.
\]
Replacing $\tau$ by $t$, the new system is that of the form:
\begin{equation}\label{eq:ns-step2}
\begin{cases}
\dot x_2=
d_{00}+d_{10}x_2+y_2+d_{20}x_2^2
+d_{30}x_2^3+d_{40}x_2^4
+O(|\mu|x_2^5),\\
\dot y_2=
y_2(e_{01}+e_{11}x_2+e_{21}x_2^2+e_{31}x_2^3+O(x_2^4)).
\end{cases}
\end{equation}
where $d_{i0}=O(|\mu|)$, $i=3,4,\ldots$. By smooth preparation, we may write
\[
\dot x_2-y_2
=
R(x_2,\mu)
(f_{00}+f_{10}x_2+f_{20}x_2^2),
\]
where $R$ is smooth, $R(x_2,0)\equiv1$, and the preparation is chosen so that $f_{00}=d_{00}+O(\mu^2)$, $f_{10}=d_{10}+O(\mu^2)$, and $f_{20}=d_{20}+O(\mu^2)$. In particular,
\[
R(0,0)=1.
\]
We then apply the transformations:
\[
x_3=x_2,\qquad
y_3=\frac{y_2}{R(x_2,\mu)},\qquad
\tau=\int_0^t R(x_2(s),\mu)\,ds,
\]
and replace $\tau$ by $t$. This removes the higher-order terms from $\dot x_3$ and gives:
\begin{equation}\label{eq:ns-step3}
\begin{cases}
\dot x_3=f_{00}+f_{10}x_3+f_{20}x_3^2+y_3,\\
\dot y_3=
\frac{R(x_2,\mu)\dot y_2-y_2R'(x_2,\mu)\dot x_2}
{R^3(x_2,\mu)}\\
\hspace{1cm}
=
y_3\left(
g_{01}+g_{11}x_3+g_{21}x_3^2+g_{31}x_3^3
+y_3O(|x_3,\mu|)+O(x_3^4)
\right).
\end{cases}
\end{equation}
At $\mu=0$,
\[
f_{20}(0)
=
-\frac{3(bh+1)^2}{8h^3(2bh+3)}
\ne0
\]
for $-1<bh<0$. Next, set
\[
z_3=x_3+\frac{f_{10}}{2f_{20}},
\]
which eliminates the linear term in the first equation. After this shift, relabel the nonzero coefficient of $z_3y_3$ in the second equation as $g_{11}$. Direct expansion gives
\[
g_{11}(0)=\frac{2bh+3}{6h(bh+1)}\ne0.
\]
The rescaling
\[
x_4=g_{11}z_3,\qquad y_4=g_{11}y_3
\]
then normalizes that coefficient to unity while preserving the unit coefficient of $y_4$ in the first equation. In the new coordinates $(x_4,y_4)$, the system is reduced to the following normal form:
\begin{equation}\label{eq:ns-step4}
\begin{cases}
\dot x_4=h_{00}+y_4+h_{20}x_4^2,\\
\dot y_4=
y_4(b_{01}+x_4+b_{21}x_4^2+b_{31}x_4^3+b_{41}x_4^4+O(x_4^5))
+y_4^2O(|x_4,\mu|).
\end{cases}
\end{equation}
A direct computation at $\mu=0$ gives
\[
b_{41}(0)
=
\frac{3bh(2b^2h^2+9bh+12)}{2(2bh+3)^3}.
\]
Since $-1<bh<0$,
\[
2b^2h^2+9bh+12>0,
\]
and hence
\[
b_{41}(0)\ne0.
\]
Therefore $b_{41}(\mu)\ne0$ for sufficiently small $\mu$. Finally, using the transformation:
\[
x_5=|b_{41}(\mu)|^{\frac{1}{3}}x_4,\qquad
y_5=|b_{41}(\mu)|^{\frac{2}{3}}y_4,\qquad
\tau=|b_{41}(\mu)|^{-\frac{1}{3}}t,
\]
the coefficient $b_{41}$ is transformed to $\pm1$, and the system takes its final form:
\begin{equation}
\begin{cases}
\dot x_5=y_5+\delta_1(\mu)+h_{20}(\mu)x_5^2,\\
\dot y_5=
y_5\left(
\delta_2(\mu)+x_5+\delta_3(\mu)x_5^2
+\omega(\mu)x_5^3\pm x_5^4+O(x_5^5)
\right)
+y_5^2O(|x_5,\mu|).
\end{cases}
\end{equation}
where
\[
\delta_1(\mu)
=
h_{00}(\mu)|b_{41}(\mu)|^{\frac{2}{3}},
\qquad
\delta_2(\mu)
=
b_{01}(\mu)|b_{41}(\mu)|^{\frac{1}{3}},
\]
\[
\delta_3(\mu)
=
b_{21}(\mu)|b_{41}(\mu)|^{-\frac{1}{3}},
\qquad
\omega(\mu)
=
b_{31}(\mu)|b_{41}(\mu)|^{-\frac{2}{3}}.
\]
Moreover, the parameter correspondence
\[
(\mu_1,\mu_2,\mu_3)
\longrightarrow
(\delta_1(\mu),\delta_2(\mu),\delta_3(\mu))
\]
is a local diffeomorphism in a neighborhood of $\mu=(0,0,0)$. Indeed, direct expansion of system~\eqref{Rescaled Main Model} through the preceding transformations gives
\[
\left.
\Det\!\left(
\frac{\partial(h_{00},b_{01},b_{21})}
{\partial(\mu_1,\mu_2,\mu_3)}
\right)
\right|_{\mu=0}
=
-\frac{(bh+1)(9(bh+1)^3-4h^2(2bh+3)^2)}
{64h^5(2bh+3)^2}.
\]
Moreover,
\[
h_{00}(0)=b_{01}(0)=b_{21}(0)=0.
\]
Therefore,
\[
\left.
\Det\!\left(
\frac{\partial(\delta_1,\delta_2,\delta_3)}
{\partial(\mu_1,\mu_2,\mu_3)}
\right)
\right|_{\mu=0}
=
|b_{41}(0)|^{\frac{2}{3}}
\left.
\Det\!\left(
\frac{\partial(h_{00},b_{01},b_{21})}
{\partial(\mu_1,\mu_2,\mu_3)}
\right)
\right|_{\mu=0}.
\]
Hence
\[
\left.
\Det\!\left(
\frac{\partial(\delta_1,\delta_2,\delta_3)}
{\partial(\mu_1,\mu_2,\mu_3)}
\right)
\right|_{\mu=0}
=
-|b_{41}(0)|^{\frac{2}{3}}
\frac{(bh+1)(9(bh+1)^3-4h^2(2bh+3)^2)}
{64h^5(2bh+3)^2}
\ne0
\]
by~\eqref{NS nondegeneracy condition}. Thus the parameter correspondence is a local diffeomorphism near the origin in parameter space, which establishes the required three-parameter versal unfolding.
\end{proof}

Having constructed a versal unfolding of the nilpotent saddle and verified that the parameter map from the original system to the unfolding parameters is a local diffeomorphism, we can transfer the bifurcation organization obtained in the normal form back to the full predator-prey model. In particular, the diffeomorphism guarantees that the codimension count is preserved and that the local bifurcation strata identified in the unfolding space pull back to corresponding strata in the original parameter space, with the same regularity. For heteroclinic loops contained in the normal-form neighborhood and covered by the cited unfolding results, the predicted locus corresponds to a local bifurcation set of the original system, and the corresponding nearby qualitative scenarios (creation or destruction of limit cycles and the emergence of these connections) persist under the inverse coordinate and parameter changes (Fig. \ref{fig:Heteroclinic loop}). The nilpotent saddle therefore organizes the nearby heteroclinic structures and adjacent bifurcations in the biologically admissible parameter region.

\begin{figure}[!htbp]
    \centering
    \includegraphics[width=0.5\linewidth]{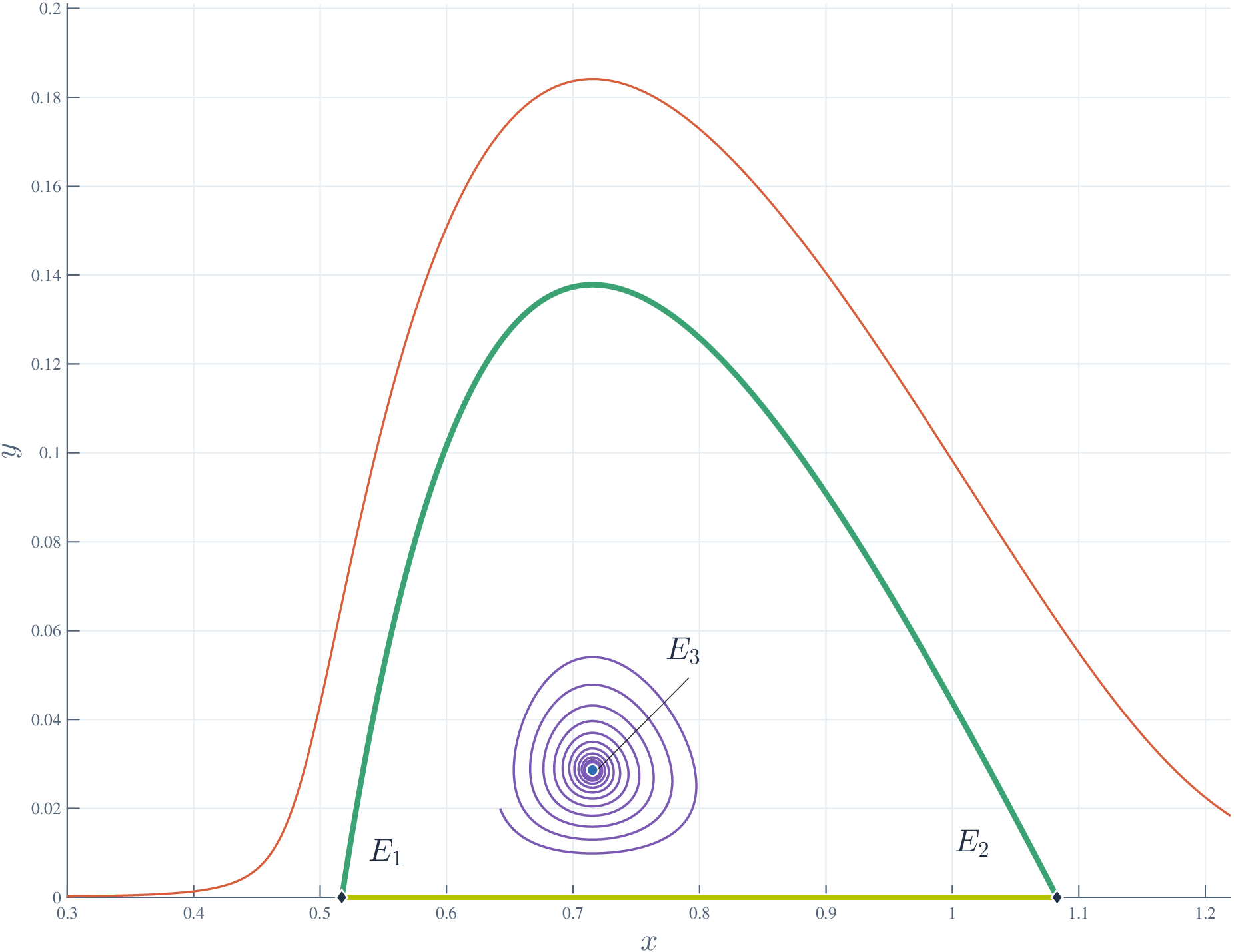}
    \caption{The formation of the Heteroclinic loop near the Nilpotent Saddle bifurcation point for
    $d\approx 1.59;\ K=1.6;\ h=0.35;\ b=-1.76755;\ a=1.145514;$}
    \label{fig:Heteroclinic loop}
\end{figure}

\FloatBarrier
\subsection{Transcritical and Transcritical--Saddle-Node Bifurcations on the x-axis}

Next, we study bifurcations near the $x$-axis.
\begin{theorem}\label{thm:boundary-collisions}
Assume $K>4h$ and let $i=1,2$. The predator-free equilibrium $E_i=(x_{b_i},0)$ undergoes a transcritical bifurcation when
\[
d=\frac{x_{b_i}^2}{ax_{b_i}^2+bx_{b_i}+1},
\qquad b\ne-\frac{2}{x_{b_i}},
\]
corresponding to its collision with precisely one of $E_3$ and $E_4$. If
\[
b=-\frac{2}{x_{b_i}},
\qquad
 d=\frac{x_{b_i}^2}{ax_{b_i}^2-1},
\]
then $E_i$, $E_3$, and $E_4$ coincide at $E_{134}$ for $i=1$, or at $E_{234}$ for $i=2$, and the system undergoes a codimension-$2$ transcritical--saddle-node (TSN) bifurcation. The reduced vector field has a nonzero negative cubic term. In particular, $E_{134}$ is of saddle type, whereas $E_{234}$ is asymptotically stable.
\end{theorem}

\begin{proof}
At $E_i$, $i=1,2$, the eigenvalues satisfy
\[
\lambda_1^i=\frac{K-2x_{b_i}}{K},
\qquad
\lambda_2^i=\frac{x_{b_i}^2}{ax_{b_i}^2+bx_{b_i}+1}-d.
\]
Since $K>4h$, the two predator-free equilibria are distinct, with $\lambda_1^1>0$ and $\lambda_1^2<0$. Thus $\lambda_1^i\ne0$, and a boundary collision occurs when $\lambda_2^i=0$, namely when
\[
d=\frac{x_{b_i}^2}{ax_{b_i}^2+bx_{b_i}+1}.
\]
We translate $E_i$ to the origin by $x_1=x-x_{b_i}$ and $y_1=y$, and apply the linear transformation
\begin{equation}\label{Jordan Canonical form E13}
\begin{bmatrix}
x_1\\y_1
\end{bmatrix}
=
\begin{bmatrix}
1&
\frac{x_{b_i}^2K}
{(ax_{b_i}^2+bx_{b_i}+1)(K-2x_{b_i})}\\
0&1
\end{bmatrix}
\begin{bmatrix}
x_2\\y_2
\end{bmatrix}.
\end{equation}
The system becomes
\begin{equation}\label{E13 or E14 model centered at zero}
\begin{cases}
\dot x_2=\frac{K-2x_{b_i}}{K}x_2+O(|x_2,y_2|^2),\\
\dot y_2=
\frac{x_{b_i}(bx_{b_i}+2)}
{(ax_{b_i}^2+bx_{b_i}+1)^2}x_2y_2\\
\hspace{1cm}
+\frac{K(bx_{b_i}+2)x_{b_i}^3}
{(K-2x_{b_i})(ax_{b_i}^2+bx_{b_i}+1)^3}y_2^2
+O(|x_2,y_2|^3).
\end{cases}
\end{equation}
For $b>-\frac{2}{x_{b_i}}$, the boundary equilibrium collides with $E_3$, while for $b<-\frac{2}{x_{b_i}}$ it collides with $E_4$. When $b\ne-\frac{2}{x_{b_i}}$, the coefficient of $y_2^2$ is nonzero, and $\frac{\partial\lambda_2^i}{\partial d}=-1\ne0$. Since the boundary equilibrium persists as $d$ varies and a single coexistence equilibrium crosses the invariant $x$-axis, the collision is transcritical. Its transverse direction is repelling for $i=1$ because $\lambda_1^1>0$, and attracting for $i=2$ because $\lambda_1^2<0$.

Now suppose that
\[
b=-\frac{2}{x_{b_i}},
\qquad
 d=\frac{x_{b_i}^2}{ax_{b_i}^2-1}.
\]
Then $E_3$ and $E_4$ also merge with $E_i$; the condition $-2\sqrt a<b<0$ implies $ax_{b_i}^2-1>0$. The quadratic terms in the second equation of \eqref{E13 or E14 model centered at zero} vanish. The one-dimensional center manifold is tangent to the $y_2$-axis and can be written as $x_2=O(y_2^2)$. Its reduced dynamics satisfy
\begin{equation}\label{eq:SNT-cubic-reduction}
\dot y_2=-\frac{K^2x_{b_i}^4}
{(ax_{b_i}^2-1)^4(K-2x_{b_i})^2}y_2^3
+O(y_2^4).
\end{equation}
The coefficient of $y_2^3$ is strictly negative. Therefore trajectories on the center manifold approach the equilibrium. Since $\lambda_1^1>0$ and $\lambda_1^2<0$, the corresponding triple collision is of saddle type at $E_{134}$ and asymptotically stable at $E_{234}$.

To determine the codimension of this bifurcation,
we fix $a,K,h$ and introduce the perturbations
\[
\bar{\epsilon}_1=b+\frac{2}{x_{b_i}},
\qquad
\bar{\epsilon}_2=d-\frac{x_{b_i}^2}{ax_{b_i}^2-1}.
\]
Since $K$ and $h$ remain fixed, the predator-free
equilibrium $(x_{b_i},0)$ does not change.
Near this equilibrium, with $y=y_2$, we obtain
\[
x=x_{b_i}
+\frac{Kx_{b_i}^2}
{(ax_{b_i}^2+bx_{b_i}+1)
(K-2x_{b_i}+Kd)-Kx_{b_i}^2}
y_2+O(y_2^2).
\]

Substituting this expansion for $x$ into the
predator equation
\[
\dot y=y\left(
-d+\frac{x^2}{ax^2+bx+1}
\right)
\]
gives the following one-dimensional reduced equation:
\begin{equation}
\label{eq:SNT-full-parameter-expansion}
\begin{aligned}
\dot y_2={}&
\left(
\frac{x_{b_i}^2}
{ax_{b_i}^2+bx_{b_i}+1}-d
\right)y_2
\\[3pt]
&+
\frac{Kx_{b_i}^3(bx_{b_i}+2)}
{(ax_{b_i}^2+bx_{b_i}+1)^2
\left[
(ax_{b_i}^2+bx_{b_i}+1)
(K-2x_{b_i}+Kd)-Kx_{b_i}^2
\right]}
y_2^2
\\[3pt]
&-
\left[
\frac{K^2x_{b_i}^4}
{(ax_{b_i}^2-1)^4(K-2x_{b_i})^2}
+O\left(
|\bar{\epsilon}_1|+|\bar{\epsilon}_2|
\right)
\right]y_2^3
+O(y_2^4).
\end{aligned}
\end{equation}

At $\bar{\epsilon}_1=\bar{\epsilon}_2=0$,
the coefficients of $y_2$ and $y_2^2$
both vanish, while the cubic coefficient
is strictly negative.

To verify that the degeneracy has codimension $2$,
we differentiate the coefficients of $y_2$
and $y_2^2$ with respect to
$\bar{\epsilon}_1$ and $\bar{\epsilon}_2$.
The resulting Jacobian matrix, evaluated at
$\bar{\epsilon}_1=\bar{\epsilon}_2=0$, is
\[
\begin{pmatrix}
-\frac{x_{b_i}^3}{(ax_{b_i}^2-1)^2}
&
-1
\\[8pt]
\frac{Kx_{b_i}^4}
{(ax_{b_i}^2-1)^3(K-2x_{b_i})}
&
0
\end{pmatrix}.
\]
Its determinant satisfies
\[
\frac{Kx_{b_i}^4}
{(ax_{b_i}^2-1)^3(K-2x_{b_i})}
\ne0.
\]
Thus, the two perturbations independently
control the linear and quadratic terms.
Together with the nonzero cubic coefficient,
this establishes a codimension-$2$
transcritical--saddle-node (TSN) bifurcation
\cite{saputra2010snt}.

The corresponding bifurcation curves are
\[
d=\frac{x_{b_i}^2}
{ax_{b_i}^2+bx_{b_i}+1}
\qquad
\text{(transcritical bifurcation)}
\]
and
\[
d=\frac{4}{4a-b^2}
\qquad
\text{(saddle-node bifurcation)}.
\]
These curves meet tangentially, since
\[
\frac{4}{4a-b^2}
-
\frac{x_{b_i}^2}
{ax_{b_i}^2+bx_{b_i}+1}
=
\frac{(bx_{b_i}+2)^2}
{(4a-b^2)(ax_{b_i}^2+bx_{b_i}+1)}.
\]
They intersect at
\[
b=-\frac{2}{x_{b_i}},
\qquad
d=\frac{x_{b_i}^2}{ax_{b_i}^2-1},
\]
where the three equilibria merge.
This completes the proof.
\end{proof}

\FloatBarrier

Fig. \ref{fig:TSN-unfolding-bd} illustrates the saddle-node and transcritical bifurcation curves in the $(b,d)$-plane. Here $d_{\mathrm{fold}}(b):=\frac{4}{4a-b^2}$ denotes the saddle-node curve used on the vertical axis of the left panel. The curves meet tangentially at $\mathrm{TSN}_{134}$ and $\mathrm{TSN}_{234}$, where three equilibria coincide.

\begin{figure}[!htbp]
    \centering
    \includegraphics[width=\linewidth]{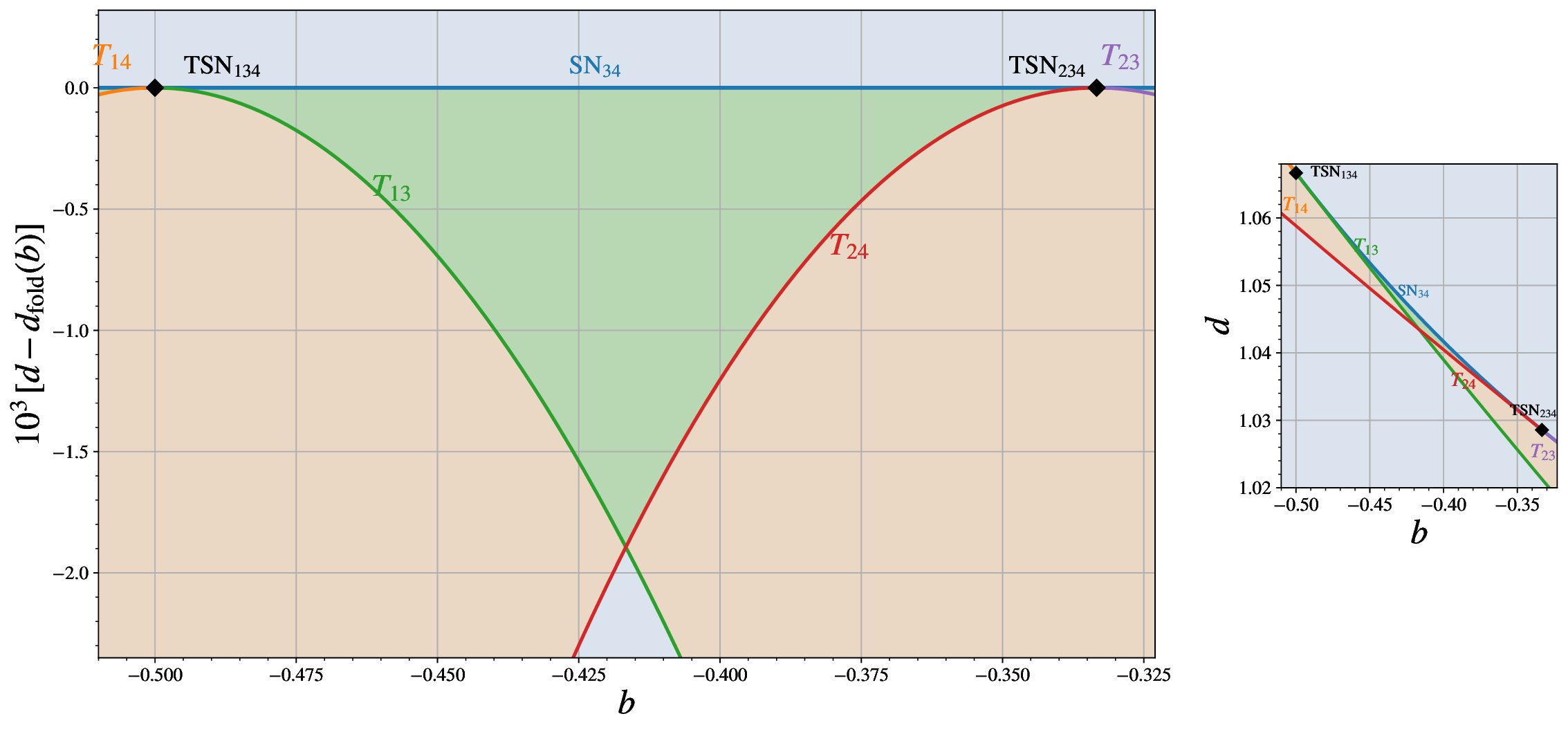}
    \caption{Transcritical--saddle-node unfolding in the $(b,d)$-plane for $a=1$, $K=10$, and $h=2.4$. The left panel magnifies the separation from the saddle-node curve, and the right panel shows the original parameter plane. $T_{ij}$ denotes a transcritical collision of $E_i$ and $E_j$, $\mathrm{SN}_{34}$ the saddle-node collision of $E_3$ and $E_4$, and the black diamonds $\mathrm{TSN}_{134}$ and $\mathrm{TSN}_{234}$ the triple collisions. Green, beige, and blue-gray indicate two, one, and zero feasible coexistence equilibria, respectively.}
    \label{fig:TSN-unfolding-bd}
\end{figure}

\FloatBarrier
\section{Discussion}

\begin{table}[!htbp]
\centering
\small
\renewcommand{\arraystretch}{1.25}
\setlength{\tabcolsep}{4pt}
\begin{tabularx}{0.95\linewidth}{|c|L|L|}
\hline
\textbf{Equilibrium} & \textbf{Bifurcation type} & \textbf{Conditions}\\
\hline

\multirow{5}{*}{$E_{34}$}
 & Attracting semi-hyperbolic saddle-node
 & \cond{d=d_0=\frac{4}{4a-b^2},\quad -\frac{2}{b}<K<K_0=-\frac{4}{b},\quad 0<h<\frac{4+2bK}{-b^2K}}
 \\ \cline{2-3}

 & Repelling semi-hyperbolic saddle-node
 & \cond{d=d_0,\quad K>K_0,\quad 0<h<\frac{4+2bK}{-b^2K}}
 \\ \cline{2-3}

 & Nilpotent cusp (codimension-2)
 & \cond{d=d_0,\quad K=K_0,\quad 0<h<-\frac{1}{b},\quad h\ne h_0}
 \\ \cline{2-3}

 & Nilpotent cusp (codimension-3)
 & \cond{d=d_0,\quad K=K_0,\quad h=h_0=\frac{-2b^2+4a}{b^3},\quad 2a<b^2<4a}
 \\ \cline{2-3}

 & Bogdanov-Takens bifurcation
 & \cond{d=d_0,\quad K=K_0,\quad h=h_0,\quad 2a<b^2<4a}\\
\hline

$E_3$ & Weak focus; Hopf bifurcation of codimension-2
 & \cond{B_1=B_2=B_3=B_4=0,\quad B_5\ne0,\quad \operatorname{rank}D(B_1,B_3)=2}\\
\hline

\multirow{2}{*}{$E_{12}$}
 & Attracting semi-hyperbolic saddle-node
 & \cond{K=K_2=4h,\quad (-4ad+4)h^2-2hdb-d<0}
 \\ \cline{2-3}

 & Repelling semi-hyperbolic saddle-node
 & \cond{K=K_2=4h,\quad (-4ad+4)h^2-2hdb-d>0}\\
\hline

$E_{13}$ & Transcritical collision; repelling transverse direction
 & \cond{d=p(x_{b_1}),\quad b>-\frac{2}{x_{b_1}}}\\
\hline

$E_{14}$ & Transcritical collision; repelling transverse direction
 & \cond{d=p(x_{b_1}),\quad b<-\frac{2}{x_{b_1}}}\\
\hline

$E_{134}$ & Semi-hyperbolic saddle-type equilibrium
 & \cond{d=p(x_{b_1}),\quad b=-\frac{2}{x_{b_1}}}\\
\hline

$E_{23}$ & Transcritical collision; attracting transverse direction
 & \cond{d=p(x_{b_2}),\quad b>-\frac{2}{x_{b_2}}}\\
\hline

$E_{24}$ & Transcritical collision; attracting transverse direction
 & \cond{d=p(x_{b_2}),\quad b<-\frac{2}{x_{b_2}}}\\
\hline

$E_{234}$ & Semi-hyperbolic stable node
 & \cond{d=p(x_{b_2}),\quad b=-\frac{2}{x_{b_2}}}\\
\hline

$E_{123}$ & Nilpotent saddle of codimension-3
 & \cond{d=\frac{4bh^3+6h^2}{3(bh+1)^2},\quad K=4h,\quad
a=\frac{2b^2h^2+4bh+3}{4h^2(2bh+3)},\quad -1<bh<0}\\
\hline

\multirow{2}{*}{$E_{124}$}
 & Nilpotent elliptic point of codimension-2
 & \cond{K=4h,\quad d=p(2h),\quad -\frac{3}{2}<bh<-1,\quad a\ne a_2}
 \\ \cline{2-3}

 & Degenerate nilpotent elliptic point
 & \cond{K=4h,\quad d=p(2h),\quad -\frac{3}{2}<bh<-1,\quad a=a_2}\\
\hline

\end{tabularx}
\caption{Summary of the local bifurcation types and the associated parameter conditions for equilibria and equilibrium coalescences in system~\eqref{Rescaled Main Model}. For the separate boundary collisions, $K>4h$ is assumed. $E_{ijk}$ denotes the coalescence point of $E_i$, $E_j$, and $E_k$, where $i,j,k\in\{1,2,3,4\}$}
\label{bifurcation-summary}
\end{table}

The biological consequences of harvesting depend strongly on how predation efficiency changes with prey density. Table~\ref{tab:differences between b>0 and b<0} summarizes the main differences between the two functional response regimes. When $b>0$, predation efficiency increases with prey density before approaching saturation. In contrast, for $b<0$, predation efficiency reaches a maximum at an intermediate prey density and then declines, reflecting the possible benefits of group defense at high prey abundance. This difference allows a second coexistence equilibrium $E_4$ to emerge. Since $E_4$ is a saddle, it can separate populations that recover toward coexistence from those that experience predator extinction. Consequently, the local stability of a coexistence equilibrium does not necessarily guarantee resilience to population disturbances.

\begin{table}[!htbp]
\centering
\renewcommand{\arraystretch}{1.25}
\begin{tabular}{lcc}
\toprule
\textbf{Dynamical feature} & \(b>0\) & \(b<0\)\\
\midrule
Maximum number of biologically feasible equilibria & 3 & 4\\
Hopf bifurcation at coexistence & $\checkmark$ & $\checkmark$\\
Codimension-two generalized Hopf bifurcation & $\checkmark$ & $\checkmark$\\
Codimension-three nilpotent-saddle bifurcation & $\checkmark$ & $\checkmark$\\
Codimension-two transcritical--saddle-node (TSN) bifurcation & -- & $\checkmark$\\
Interior codimension-three BT bifurcation at $E_{34}$ & -- & $\checkmark$\\
Two harvesting thresholds for predator invasion & -- & $\checkmark$\\
\bottomrule
\end{tabular}
\caption{Comparison of the two response regimes. A check indicates that the phenomenon occurs for suitable parameter values; a dash indicates that it does not arise through the mechanism considered here.}
\label{tab:differences between b>0 and b<0}
\end{table}

These differences become particularly relevant when considering predator establishment under harvesting. For $b<0$, suppose
\[
\frac{1}{a}<d<\frac{4}{4a-b^2},
\qquad
\frac{K}{2}<x_\ell<x_r<K,
\]
where $p(x_\ell)=p(x_r)=d$. The corresponding harvesting thresholds are
\[
h_r=x_r\left(1-\frac{x_r}{K}\right),
\qquad
h_\ell=x_\ell\left(1-\frac{x_\ell}{K}\right),
\]
with
\[
0<h_r<h_\ell<\frac{K}{4}.
\]
As harvesting increases, the prey density at the larger predator-free equilibrium $E_2$ decreases. Since $p(x)>d$ precisely when $x_\ell<x<x_r$, predators can invade when rare if and only if
\[
h_r<h<h_\ell.
\]
These thresholds correspond to the transcritical collisions described in Theorem~\ref{thm:boundary-collisions}. At $h=h_r$, we have $x_{b_2}=x_r$, so $E_2$ meets $E_4$. At $h=h_\ell$, we have $x_{b_2}=x_\ell$, so $E_2$ meets $E_3$. In both cases, the predator-direction eigenvalue
\[
\lambda_2^2=p(x_{b_2})-d
=\frac{x_{b_2}^2}{ax_{b_2}^2+bx_{b_2}+1}-d
\]
is zero, while $\lambda_1^2=1-\frac{2x_{b_2}}{K}<0$. Therefore, $\lambda_2^2>0$ when $h_r<h<h_\ell$ and $\lambda_2^2<0$ outside this interval for $0<h<\frac{K}{4}$. The two harvesting thresholds thus mark where predators gain or lose the ability to invade when rare.

Fig. \ref{fig:harvesting_comparison} illustrates how harvesting affects predator establishment in the two response regimes. The curves in the upper part of each panel show the growth rate of a rare predator, $p(x_{b_2})-d$: positive values allow invasion, while negative values prevent it. The basin diagrams below show which initial population densities lead to coexistence (green), predator extinction (blue), or prey extinction (red). The three harvesting levels have the following interpretations:
\begin{itemize}
\item \textbf{Low harvesting ($h=0.04$).} In the $b<0$ panel (top), prey are abundant and group defense reduces their vulnerability. The invasion growth rate is negative, so predators introduced at low density cannot establish. Nevertheless, the green region shows that coexistence remains possible when the initial predator population is sufficiently large; the blue region leads to predator extinction. For $b>0$ (bottom), the invasion growth rate is positive and coexistence occupies most of the displayed basin.

\item \textbf{Moderate harvesting ($h=0.30$).} In the $b<0$ panel, harvesting lowers prey abundance and weakens group defense. The invasion growth rate becomes positive between the thresholds $h_r$ and $h_\ell$, allowing rare predators to increase. The green region shows coexistence for a substantial range of initial conditions, although the red region indicates that some initial states lead to prey extinction. For $b>0$, the invasion growth rate remains positive but has declined from its low-harvesting value.

\item \textbf{High harvesting ($h=0.49$).} In the $b<0$ panel, harvesting has reduced prey abundance below the level needed for positive predator growth, and the invasion growth rate becomes negative again beyond $h_\ell$. The basin diagram now shows prey extinction (red) or predator extinction (blue), with no coexistence region at the illustrated harvesting level. The $b>0$ panel shows the same two extinction outcomes after its single invasion threshold has been crossed.
\end{itemize}
Thus, for $b<0$, moderate harvesting can facilitate predator invasion by weakening group defense, whereas the invasion potential for $b>0$ declines as harvesting increases. The basin diagrams also show that invasion by a rare predator and long-term persistence are distinct questions: the eventual outcome can depend on the initial prey and predator densities, even at a fixed harvesting rate.

\begin{figure}[!htbp]
    \centering

    \begin{minipage}[c]{0.015\linewidth}
        \centering
        \rotatebox[origin=c]{90}{\large $b<0$}
    \end{minipage}    \hspace{0.005\linewidth}    \begin{minipage}[c]{0.965\linewidth}
        \centering
        \includegraphics[width=0.8\linewidth]{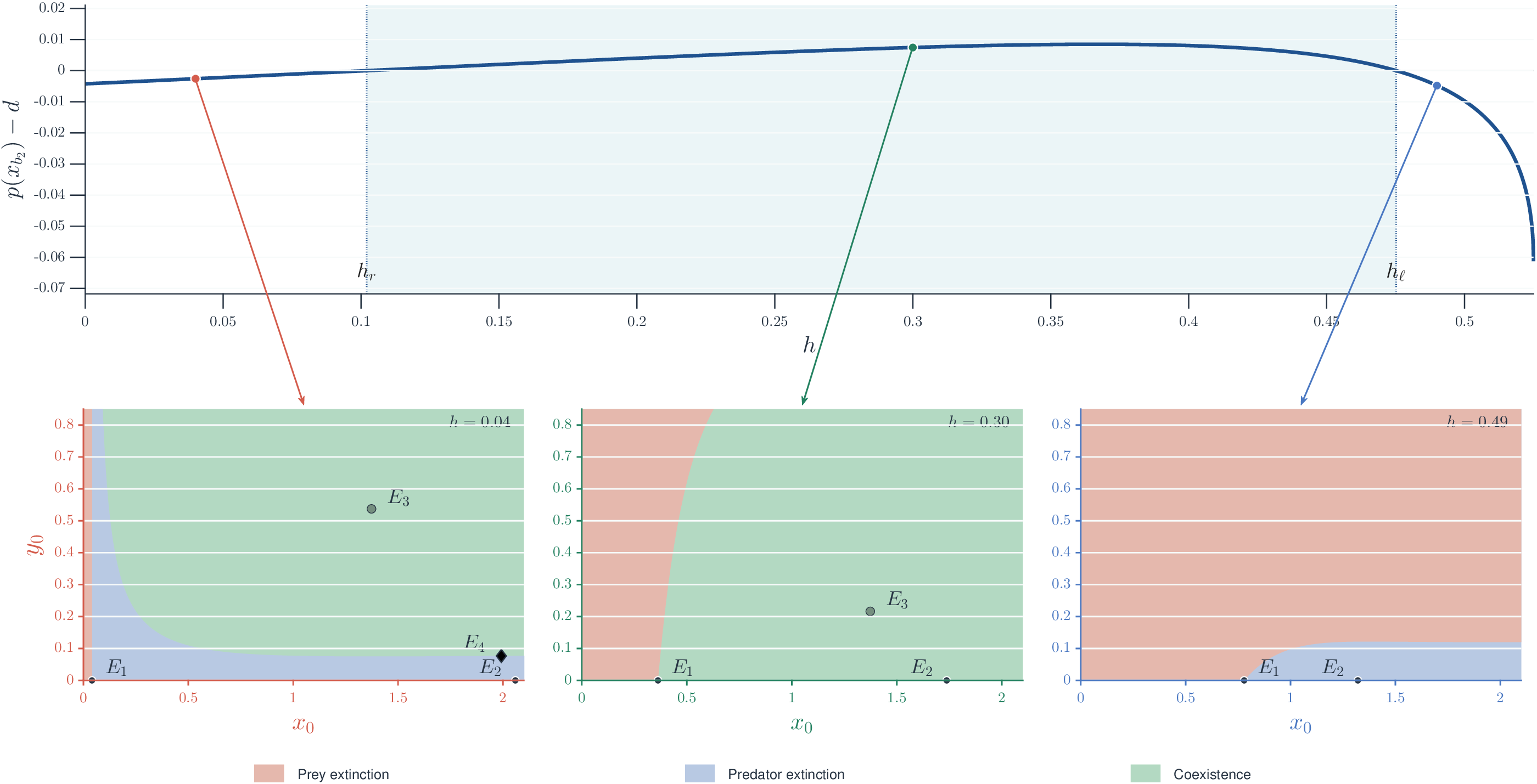}
    \end{minipage}

    \vspace{0.5em}

    \begin{minipage}[c]{0.015\linewidth}
        \centering
        \rotatebox[origin=c]{90}{\large $b>0$}
    \end{minipage}    \hspace{0.005\linewidth}    \begin{minipage}[c]{0.965\linewidth}
        \centering
        \includegraphics[width=0.8\linewidth]{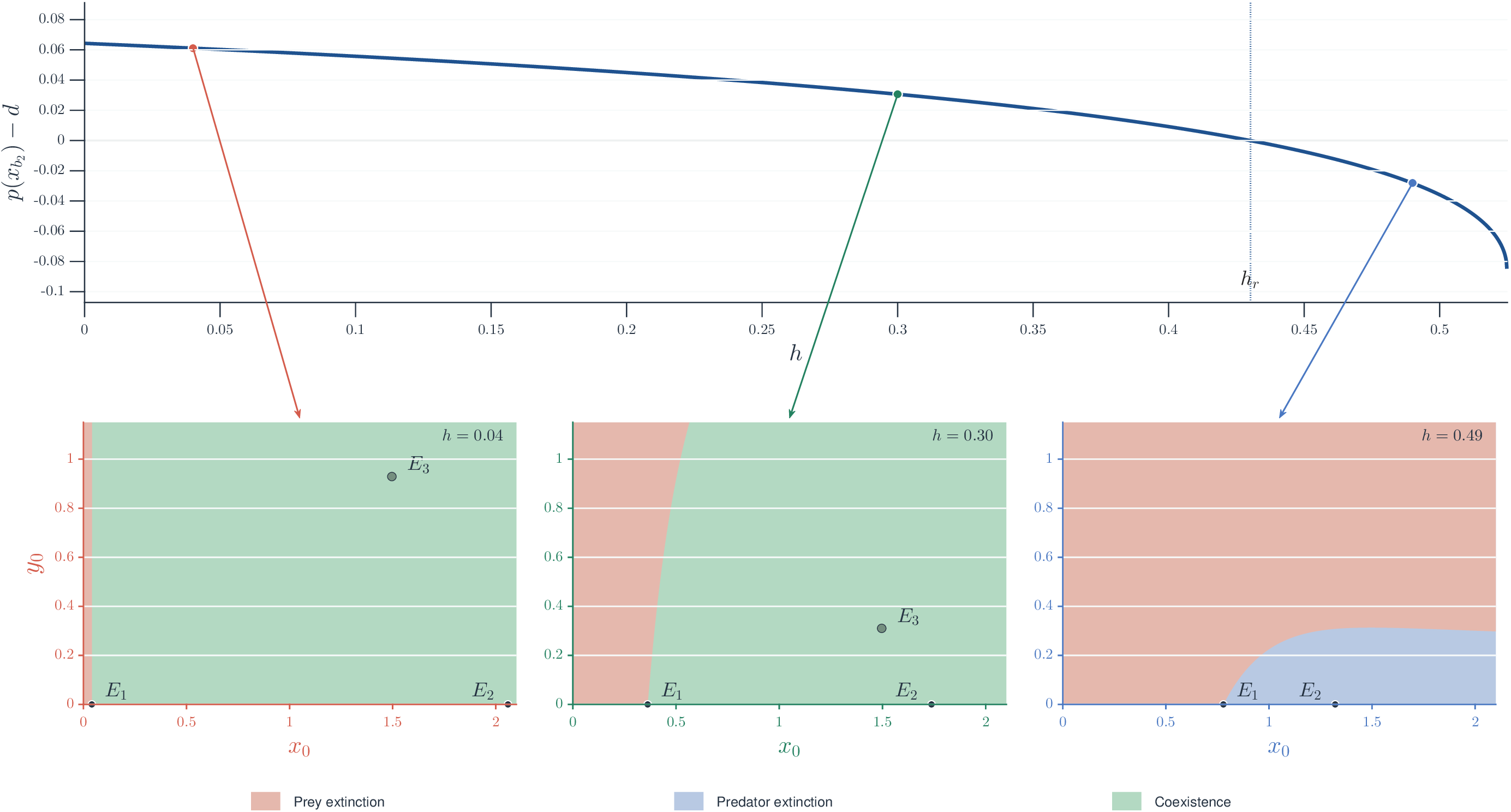}
    \end{minipage}

    \caption{
    Effects of harvesting on predator invasion and ecosystem persistence.
    Parameters $(a,b,d,K)=(1.6,-1.23,0.81,2.1)$ (top) and
    $(1.6,0.5,0.42,2.1)$ (bottom).
    Basin diagrams correspond to $h=0.04$, $0.30$, and $0.49$. Prey extinction occur in finite time, as the orbits leave the first quadrant through the positive y-axis.
    }
    \label{fig:harvesting_comparison}
\end{figure}

\FloatBarrier

Beyond predator establishment, the sign of $b$ also affects the parameter values at which population oscillations and degenerate equilibria arise. Although Hopf bifurcations occur in both regimes, their critical harvesting levels depend on the shape of the functional response. At a feasible Hopf bifurcation of $E_3=(x_\ell,y_\ell)$, the harvesting threshold satisfies
\[
h_H=x_\ell\left(1-\frac{x_\ell}{K}\right)
-\frac{d}{p'(x_\ell)}
\left(1-\frac{2x_\ell}{K}\right).
\]
Allowing $b<0$ changes both the coexistence prey density $x_\ell$ and the slope $p'(x_\ell)$, thereby shifting the harvesting level at which steady coexistence can give way to population oscillations. Depending on the remaining parameters, these oscillations may arise at lower or higher harvesting intensities. Hence, the effect of group defense extends beyond predator invasion and also influences the conditions under which predator and prey populations fluctuate over time.

A similar shift occurs in the nilpotent-saddle bifurcation, where the predator-free equilibria merge with a coexistence equilibrium. The corresponding thresholds satisfy
\[
K=4h,\qquad
d=p(2h)=\frac{4h^2}{4ah^2+2bh+1}.
\]
For fixed $a$ and $h$, decreasing $b$ increases the critical predator mortality rate, provided the functional response remains well defined. Thus, the change in predation efficiency modifies the conditions under which equilibria merge and nearby global connections or population cycles may arise. Furthermore, when $b<0$, the additional coexistence equilibrium allows an interior saddle-node bifurcation that can develop into the codimension-three Bogdanov--Takens structure established in this paper.

Overall, these results show that the ecological consequences of harvesting depend not only on the number of prey removed but also on how prey abundance influences predation efficiency. In particular, group defense can create an intermediate harvesting range favorable to predator establishment, alter the resilience of coexistence to disturbances, and shift the thresholds associated with oscillatory and degenerate population dynamics.

\FloatBarrier
\newpage

\section*{Appendix A}
\begingroup
\setlength{\parindent}{0pt}
\newcounter{appBSavedEquation}
\setcounter{appBSavedEquation}{\value{equation}}
\setcounter{equation}{0}
\renewcommand{\theequation}{A.\arabic{equation}}
\renewcommand{\theHequation}{appB.\arabic{equation}}
\small
\setlength{\jot}{4pt}
\setlength{\abovedisplayskip}{4pt}
\setlength{\belowdisplayskip}{4pt}
\setlength{\abovedisplayshortskip}{4pt}
\setlength{\belowdisplayshortskip}{4pt}
\allowdisplaybreaks[3]

\noindent Throughout this appendix, we write
\(\Delta=4a-b^2\), \(U=1+bh\), and \(V=3+2bh\)
for recurring algebraic expressions. These are abbreviations only. The Bogdanov--Takens coefficient families use $\epsilon=(\epsilon_1,\epsilon_2,\epsilon_3)$, whereas the nilpotent-saddle coefficient families use $\mu=(\mu_1,\mu_2,\mu_3)$. Parameter-dependent coefficients are recorded through first order in the corresponding perturbation vector; an omitted parameter remainder is $O(\|\epsilon\|^2)$ or $O(\|\mu\|^2)$, respectively.
\par\vspace{5pt}

\noindent\textit{Coefficients $\mathrm{o}_{ij}$ (Coefficients of \eqref{BT system}).}\par\nobreak\vspace{1pt}
{\renewcommand{\arraystretch}{1.25}
\refstepcounter{equation}\label{eq:appB:o}
\noindent\begin{minipage}[c]{.87\linewidth}
\begin{tabular}{@{}p{.46\linewidth}@{\hspace{.8em}}p{.48\linewidth}@{}}
$\displaystyle \mathrm{o}_{01}=-\frac{4}{\Delta}$
&
$\displaystyle \mathrm{o}_{30}=-\frac{b^2}{4}$\\[4pt]
$\displaystyle \mathrm{o}_{21}=\frac{b^4}{\Delta^2}$
&
$\displaystyle \mathrm{o}_{40}=\frac{b^3(b^2-3a)}{4\Delta}$\\[4pt]
$\displaystyle \mathrm{o}_{31}=\frac{b^5}{\Delta^2}$
&
\end{tabular}
\end{minipage}\hfill(\theequation)\par}\vspace{4pt}

\noindent\textit{Coefficients $\iota_{ij}$ (Coefficients of \eqref{BT system}).}\par\nobreak\vspace{1pt}
{\renewcommand{\arraystretch}{1.25}
\refstepcounter{equation}\label{eq:appB:iota}
\noindent\begin{minipage}[c]{.87\linewidth}
\begin{tabular}{@{}p{.46\linewidth}@{\hspace{.8em}}p{.48\linewidth}@{}}
$\displaystyle \iota_{20}=\frac{b}{4}$
&
$\displaystyle \iota_{30}=\frac{b^2}{4}$\\[4pt]
$\displaystyle \iota_{21}=-\frac{b^4}{\Delta^2}$
&
$\displaystyle \iota_{40}=\frac{b^3(3a-b^2)}{4\Delta}$\\[4pt]
$\displaystyle \iota_{31}=-\frac{b^5}{\Delta^2}$
&
\end{tabular}
\end{minipage}\hfill(\theequation)\par}\vspace{4pt}

\noindent\textit{Coefficients $\rho_{ij}$ (Coefficients of \eqref{eq:cusp-rho-expansion}).}\par\nobreak\vspace{1pt}
{\renewcommand{\arraystretch}{1.25}
\refstepcounter{equation}\label{eq:appB:rho}
\noindent\begin{minipage}[c]{.87\linewidth}
\begin{tabular}{@{}p{.46\linewidth}@{\hspace{.8em}}p{.48\linewidth}@{}}
$\displaystyle \rho_{20}=-\frac{b}{\Delta}$
&
$\displaystyle \rho_{30}=-\frac{b^2}{\Delta}$\\[4pt]
$\displaystyle \rho_{12}=-\frac{b^4}{2\Delta}$
&
$\displaystyle \rho_{21}=-\frac{b^2(3b^4-24ab^2+48a^2+4b^2)}{4\Delta^2}$\\[4pt]
$\displaystyle \rho_{40}=\frac{b^3(5b^2-12a)}{4\Delta^2}$
&
$\displaystyle \rho_{22}=-\frac{3b^5}{4\Delta}$\\[4pt]
$\displaystyle \rho_{31}=-\frac{b^3(b^4-7ab^2+12a^2+b^2)}{\Delta^2}$
&
\end{tabular}
\end{minipage}\hfill(\theequation)\par}\vspace{4pt}

\noindent\textit{Coefficients $\varpi_{ij}$ (Coefficients of \eqref{eq:cusp-varpi-expansion}).}\par\nobreak\vspace{1pt}
{\renewcommand{\arraystretch}{1.25}
\refstepcounter{equation}\label{eq:appB:varpi}
\noindent\begin{minipage}[c]{.87\linewidth}
\begin{tabular}{@{}p{.46\linewidth}@{\hspace{.8em}}p{.48\linewidth}@{}}
$\displaystyle \varpi_{30}=-\Delta$
&
$\displaystyle \varpi_{12}=-\frac{b^2\Delta}{2}$\\[4pt]
$\displaystyle \varpi_{22}=\frac{3b^2\Delta^2}{4}$
&
$\displaystyle \varpi_{21}=-\frac{3b^4}{4}+6ab^2-12a^2-b^2$\\[4pt]
$\displaystyle \varpi_{40}=\frac{(12a-5b^2)\Delta}{4}$
&
$\displaystyle \varpi_{31}=(b^4-7ab^2+12a^2+b^2)\Delta$
\end{tabular}
\end{minipage}\hfill(\theequation)\par}\vspace{4pt}

\medskip\noindent\textbf{Transformation coefficients}\par\nobreak\vspace{2pt}

\noindent\textit{Coefficients $w_{ij}$ (Coefficients of \eqref{BT system}).}\par\nobreak\vspace{1pt}
{\renewcommand{\arraystretch}{1.25}
\refstepcounter{equation}\label{eq:appB:w}
\noindent\begin{minipage}[c]{.87\linewidth}
\begin{tabular}{@{}p{.46\linewidth}@{\hspace{.8em}}p{.48\linewidth}@{}}
$\displaystyle w_{00}=\frac{\epsilon_2}{4}-\epsilon_3+O(\epsilon^2)$
&
$\displaystyle w_{10}=-\frac{b\epsilon_2}{4}+O(\epsilon^2)$\\[4pt]
$\displaystyle w_{20}=\frac{b^2\epsilon_2}{16}+O(\epsilon^2)$
&
$\displaystyle w_{01}=\mathrm{o}_{01}$\\[4pt]
$\displaystyle w_{30}=\mathrm{o}_{30}$
&
$\displaystyle w_{21}=\mathrm{o}_{21}$\\[4pt]
$\displaystyle w_{40}=\mathrm{o}_{40}$
&
$\displaystyle w_{31}=\mathrm{o}_{31}$
\end{tabular}
\end{minipage}\hfill(\theequation)\par}\vspace{4pt}

\noindent\textit{Coefficients $k_{ij}$ (Coefficients of \eqref{BT system}).}\par\nobreak\vspace{1pt}
{\renewcommand{\arraystretch}{1.25}
\refstepcounter{equation}\label{eq:appB:k}
\noindent\begin{minipage}[c]{.87\linewidth}
\begin{tabular}{@{}p{.46\linewidth}@{\hspace{.8em}}p{.48\linewidth}@{}}
$\displaystyle k_{00}=\frac{\Delta^2\epsilon_1}{4b^3}$
&
$\displaystyle k_{01}=-\epsilon_1$\\[4pt]
$\displaystyle k_{20}=\iota_{20}$
&
$\displaystyle k_{30}=\iota_{30}$\\[4pt]
$\displaystyle k_{21}=\iota_{21}$
&
$\displaystyle k_{40}=\iota_{40}$\\[4pt]
$\displaystyle k_{31}=\iota_{31}$
&
\end{tabular}
\end{minipage}\hfill(\theequation)\par}\vspace{4pt}

\noindent\textit{Coefficients $m_{ij}$ (Coefficients of \eqref{step_1 BT3}).}\par\nobreak\vspace{1pt}
{\renewcommand{\arraystretch}{1.25}
\noindent\begin{tabular}{@{}p{.44\linewidth}@{\hspace{1.3em}}p{.48\linewidth}@{}}
$\displaystyle m_{00}=-\frac{\Delta\epsilon_1}{b^3}$
&
$\displaystyle m_{01}=-\frac{b\epsilon_2}{4}-\epsilon_1+O(\epsilon^2)$\\[4pt]
$\displaystyle m_{12}=-\frac{b^4}{2\Delta}$
&
$\displaystyle m_{22}=-\frac{3b^5}{4\Delta}$\\[4pt]
$\displaystyle m_{10}=O(\epsilon^2)$
&
\end{tabular}\par}\vspace{4pt}
\begin{flalign}
m_{11}&=\frac{ab^2\epsilon_2-b^4\epsilon_3}{2\Delta}
+O(\epsilon^2)&&\notag\\
m_{20}&=-\frac{b}{\Delta}+\frac{b\epsilon_1}{4}
+\frac{b^4\epsilon_2}{4\Delta^2}
-\frac{b^4\epsilon_3}{\Delta^2}
+O(\epsilon^2)&&\notag\\
m_{21}&=-\frac{b^2(3b^4-24ab^2+48a^2+4b^2)}{4\Delta^2}
+\frac{b^5\epsilon_2}{16\Delta}&&\notag\\
&\quad-\frac{3b^5\epsilon_3}{4\Delta}
+O(\epsilon^2)&&\notag\\
m_{30}&=-\frac{b^2}{\Delta}-\frac{b^5\epsilon_3}{\Delta^2}&&\notag\\
m_{40}&=-\frac{b^3(12a-5b^2)}{4\Delta^2}
-\frac{b^8\epsilon_2}{16\Delta^3}
-\frac{b^6(3a-b^2)\epsilon_3}{\Delta^3}
+O(\epsilon^2)&&\notag\\
m_{31}&=-\frac{b^3(b^4-7ab^2+12a^2+b^2)}{\Delta^2}
+\frac{b^6(2a-b^2)\epsilon_2}{16\Delta^2}&&\notag\\
&\quad+\frac{b^6(7b^2-24a)\epsilon_3}{8\Delta^2}
+O(\epsilon^2).&&\label{eq:appB:m}
\end{flalign}

\noindent\textit{Coefficients $n_{ij}$ (Coefficients of \eqref{eq:bt-step2}).}\par\nobreak\vspace{1pt}
{\renewcommand{\arraystretch}{1.25}
\noindent\begin{tabular}{@{}p{.44\linewidth}@{\hspace{1.3em}}p{.48\linewidth}@{}}
$\displaystyle n_{00}=m_{00}$
&
$\displaystyle n_{01}=-\epsilon_1-\frac{b\epsilon_2}{4}$\\[4pt]
$\displaystyle n_{11}=\frac{ab^2\epsilon_2-b^4\epsilon_3}{2\Delta}$
&
$\displaystyle n_{21}=m_{21}+O(\epsilon^2)$\\[4pt]
$\displaystyle n_{10}=O(\epsilon^2)$
&
\end{tabular}\par}\vspace{4pt}
\begin{flalign}
n_{20}&=-\frac{b}{\Delta}+\frac{b^4\epsilon_2}{4\Delta^2}
-\frac{b^4\epsilon_3}{\Delta^2}&&\notag\\
n_{30}&=-\frac{b^2}{\Delta}-\frac{b^2\epsilon_1}{4}
-\frac{b^5\epsilon_3}{\Delta^2}&&\notag\\
n_{40}&=-\frac{b^3(18a-7b^2)}{6\Delta^2}
-\frac{b^5\epsilon_1}{24\Delta}&&\notag\\
&\quad-\frac{b^8\epsilon_2}{24\Delta^3}
+\frac{b^6(11b^2-36a)\epsilon_3}{12\Delta^3}
+O(\epsilon^2)&&\notag\\
n_{31}&=-\frac{b^3(b^4-7ab^2+12a^2+b^2)}{\Delta^2}
+\frac{b^6(4a-3b^2)\epsilon_2}{48\Delta^2}&&\notag\\
&\quad+\frac{b^6(11b^2-36a)\epsilon_3}{12\Delta^2}
+O(\epsilon^2).&&\label{eq:appB:n}
\end{flalign}

\noindent\textit{Coefficients $p_{ij}$ (Coefficients of \eqref{ste$P_3$ BT3}).}\par\nobreak\vspace{1pt}
{\renewcommand{\arraystretch}{1.25}
\noindent\begin{tabular}{@{}p{.44\linewidth}@{\hspace{1.3em}}p{.48\linewidth}@{}}
$\displaystyle p_{00}=m_{00}$
&
$\displaystyle p_{01}=-\epsilon_1-\frac{b\epsilon_2}{4}$\\[4pt]
$\displaystyle p_{10}=\frac{\Delta\epsilon_1}{2b^2}+O(\epsilon^2)$
&
\end{tabular}\par}\vspace{4pt}
\begin{flalign}
p_{11}&=\frac{b\epsilon_1}{2}
+\frac{b^2(8a-b^2)\epsilon_2}{8\Delta}
-\frac{b^4\epsilon_3}{2\Delta}&&\notag\\
p_{20}&=-\frac{b}{\Delta}-\frac{(11b^2+36a)\epsilon_1}{80b}
+\frac{b^4\epsilon_2}{4\Delta^2}
-\frac{b^4\epsilon_3}{\Delta^2}&&\notag\\
p_{21}&=-\frac{b^2(3b^4-24ab^2+48a^2+4b^2)}{4\Delta^2}
-\frac{b^2(11b^2+36a)\epsilon_1}{80\Delta}&&\notag\\
&\quad-\frac{3b^3(52a-3b^2)\epsilon_2}{320\Delta}
-\frac{3b^5\epsilon_3}{8\Delta},&&\notag\\
p_{31}&=\frac{b^5}{4\Delta}
+\frac{b^3(3b^4-24ab^2+48a^2+4b^2)\epsilon_1}{16\Delta}&&\notag\\
&\quad+\frac{b^4(176a^3+152a^2b^2-109ab^4+15b^6+20b^4)\epsilon_2}
{80\Delta^3}&&\notag\\
&\quad-\frac{11b^6\epsilon_3}{80\Delta}
+O(\epsilon^2).&&\label{eq:appB:p}
\end{flalign}

\noindent\textit{Coefficients $q_{ij}$ (Coefficients of \eqref{BT step 4}).}\par\nobreak\vspace{1pt}
{\renewcommand{\arraystretch}{1.25}
\noindent\begin{tabular}{@{}p{.44\linewidth}@{\hspace{1.3em}}p{.48\linewidth}@{}}
$\displaystyle q_{00}=m_{00}$
&
$\displaystyle q_{10}=\frac{\Delta\epsilon_1}{2b^2}$\\[4pt]
$\displaystyle q_{20}=p_{20}$
&
$\displaystyle q_{31}=p_{31}$
\end{tabular}\par}\vspace{4pt}
\begin{flalign}
q_{01}&=\frac{3\Delta^2\epsilon_1}{4b^2}
-\frac{b\epsilon_2}{4}&&\notag\\
q_{11}&=-\frac{3\Delta^2\epsilon_1}{8b}
+\frac{b^2(8a-b^2)\epsilon_2}{8\Delta}
-\frac{b^4\epsilon_3}{2\Delta}.&&\label{eq:appB:q}
\end{flalign}

\noindent\textit{Coefficients $r_{ij}$ (Coefficients of \eqref{eq:bt-step5}).}\par\nobreak\vspace{1pt}
{\renewcommand{\arraystretch}{1.25}
\noindent\begin{tabular}{@{}p{.44\linewidth}@{\hspace{1.3em}}p{.48\linewidth}@{}}
$\displaystyle r_{00}=
\frac{2^{\frac{2}{5}}\Delta^{\frac{8}{5}}\epsilon_1}{4(-b)^{\frac{2}{5}}}$
&
$\displaystyle r_{10}=
\frac{2^{\frac{1}{5}}\Delta^{\frac{9}{5}}\epsilon_1}{4(-b)^{\frac{6}{5}}}$
\end{tabular}\par}\vspace{4pt}
\begin{flalign}
r_{01}&=
\frac{2^{\frac{3}{5}}\Delta^{\frac{2}{5}}}{8(-b)^{\frac{8}{5}}}
(3\Delta^2\epsilon_1-b^3\epsilon_2)&&\notag\\
r_{11}&=
\frac{2^{\frac{2}{5}}}{8(-b)^{\frac{12}{5}}\Delta^{\frac{2}{5}}}
\left(
3\Delta^3\epsilon_1-b^3(8a-b^2)\epsilon_2+4b^5\epsilon_3
\right).&&\label{eq:appB:r}
\end{flalign}

\medskip\noindent\textbf{Expressions}\par\nobreak\vspace{2pt}
\noindent The quantities $\mathcal D_1,\mathcal D_2,\mathcal D_3$ are retained from the weak-focus calculations;\par\nobreak\vspace{2pt}
\begin{flalign}
\mathcal D_1&=
b^2x_\ell\left(
a^2b^2x_\ell^5+11a^2bx_\ell^4+18a^2x_\ell^3
-4abx_\ell^2+4b^2x_\ell-18ax_\ell+3b
\right)&&\notag\\[-2pt]
&\qquad\times
\left((b^2+2a)x_\ell^2+6bx_\ell+6\right),&&\notag\\
\mathcal D_2&=
24\sqrt{
-\frac{
(bx_\ell+2)
\left(2abx_\ell^3+(6a-b^2)x_\ell^2-4bx_\ell-6\right)
(ax_\ell^2+bx_\ell+1)}
{bx_\ell^3\left((b^2+2a)x_\ell^2+6bx_\ell+6\right)}
},&&\notag\\
\mathcal D_3&=
(bx_\ell+2)^3(ax_\ell^2+bx_\ell+1)^2
\left(2abx_\ell^3+(6a-b^2)x_\ell^2-4bx_\ell-6\right)^2.
&&\label{eq:appB:auxiliary}
\end{flalign}

\medskip\noindent\textbf{Additional coefficient families}\par\nobreak\vspace{2pt}

\noindent\textit{Coefficients $\ell_{ij}$ (Coefficients of \eqref{eq:ns-step1}).}\par\nobreak\vspace{1pt}
{\renewcommand{\arraystretch}{1.25}
\noindent\begin{tabular}{@{}p{.44\linewidth}@{\hspace{1.3em}}p{.48\linewidth}@{}}
$\displaystyle \ell_{00}=
\frac{3U^2\mu_1}{8h^2V}+O(\mu^2)$
&
$\displaystyle \ell_{10}=
\frac{(4bh+3)U\mu_1}{16h^3V}+O(\mu^2)$\\[4pt]
$\displaystyle \ell_{30}=\frac{U}{16h^4}$
&
$\displaystyle \ell_{40}=-\frac{V}{64h^5}$
\end{tabular}\par}\vspace{4pt}
\begin{flalign}
\ell_{20}&=
-\frac{3U^2}{8h^3V}
+\frac{(2U^2+1)\mu_1}{64h^4V}
-\frac{\mu_3}{4h}
+O(\mu^2).&&\label{eq:appB:ell}
\end{flalign}

\noindent\textit{Coefficients $c_{ij}$ (Coefficients of \eqref{eq:ns-step1}).}\par\nobreak\vspace{1pt}
{\renewcommand{\arraystretch}{1.25}
\noindent\begin{tabular}{@{}p{.44\linewidth}@{\hspace{1.3em}}p{.48\linewidth}@{}}
$\displaystyle c_{11}=\frac{V}{6hU}+\frac{U\mu_2}{4h^3}$
&
$\displaystyle c_{21}=
-\frac{V^2}{24h^2U^2}-\frac{V\mu_2}{16h^4}$
\end{tabular}\par}\vspace{4pt}
\begin{flalign}
c_{01}&=
-\frac{3U^2\mu_2}{2h^2V}
-\frac{2h^2V\mu_3}{3U^2}
+O(\mu^2)&&\notag\\
c_{31}&=
\frac{2b^2h^2+7bh+6}{24h^3U^2}
+\frac{(bh+2)\mu_2}{16h^5}.
&&\label{eq:appB:c}
\end{flalign}

\noindent\textit{Coefficients $f_{ij}$ (Coefficients of \eqref{eq:ns-step3}).}\par\nobreak\vspace{1pt}
{\renewcommand{\arraystretch}{1.25}
\noindent\begin{tabular}{@{}p{.44\linewidth}@{\hspace{1.3em}}p{.48\linewidth}@{}}
$\displaystyle f_{00}=\ell_{00}$
&
$\displaystyle f_{10}=\ell_{10}$
\end{tabular}\par}\vspace{4pt}
\begin{flalign}
f_{20}&=
-\frac{3U^2}{8h^3V}
+\frac{(7b^2h^2+15bh+9)\mu_1}{96h^4V}
-\frac{\mu_3}{4h}
+O(\mu^2).&&\label{eq:appB:f}
\end{flalign}

\noindent\textit{Coefficients $g_{ij}$ (Coefficients of \eqref{eq:ns-step3}).}\par\nobreak\vspace{1pt}
{\renewcommand{\arraystretch}{1.25}
\noindent\begin{tabular}{@{}p{.44\linewidth}@{\hspace{1.3em}}p{.48\linewidth}@{}}
$\displaystyle g_{01}=c_{01}$
&
\end{tabular}\par}\vspace{4pt}
\begin{flalign}
g_{11}&=
\frac{V}{6hU}-\frac{hV^2\mu_3}{9U^3}
+O(\mu^2)&&\notag\\
g_{21}&=
-\frac{P_1\mu_1}{288h^5U^4}
+\frac{V\mu_2}{64h^4}
+\frac{VP_2\mu_3}{432h^2U^4}
+O(\mu^2)&&\notag\\
g_{31}&=
-\frac{2b^3h^3+9b^2h^2+15bh+9}{72h^3U^3}
+\frac{P_3\mu_1}{62208h^6U^5}&&\notag\\
&\quad
-\frac{b(bh+3)\mu_2}{144h^4U}
+\frac{(2b^2h^2+6bh+9)V^2\mu_3}{324hU^5}
+O(\mu^2),&&\notag\\
g_{41}&=
\frac{4b^4h^3+24b^3h^2+51b^2h+36b}{864h^3U^4}
+\frac{P_4\mu_1}{2654208h^7U^6}&&\notag\\
&\quad
-\frac{(8b^3h^3+36b^2h^2+118bh+75)\mu_2}{4096h^6U^2}
-\frac{VP_5\mu_3}{82944h^4U^6}&&\notag\\
&\quad+O(\mu^2).&&\label{eq:appB:g}
\end{flalign}

\noindent\textit{Coefficients $d_{ij}$ (Coefficients of \eqref{eq:ns-step2}).}\par\nobreak\vspace{1pt}
{\renewcommand{\arraystretch}{1.25}
\noindent\begin{tabular}{@{}p{.44\linewidth}@{\hspace{1.3em}}p{.48\linewidth}@{}}
$\displaystyle d_{00}=\ell_{00}$
&
$\displaystyle d_{10}=\ell_{10}$\\[4pt]
$\displaystyle d_{20}=f_{20}$
&
\end{tabular}\par}\vspace{4pt}
\begin{flalign}
d_{30}&=
\frac{(16b^2h^2+30bh+27)\mu_1}{4608h^5U}
-\frac{V\mu_3}{24h^2U}
+O(\mu^2)&&\notag\\
d_{40}&=
-\frac{(8b^3h^3+30b^2h^2+54bh+27)\mu_1}{13824h^6U^2}
+O(\mu^2).&&\label{eq:appB:d}
\end{flalign}

\noindent\textit{Coefficients $e_{ij}$ (Coefficients of \eqref{eq:ns-step2}).}\par\nobreak\vspace{1pt}
{\renewcommand{\arraystretch}{1.25}
\noindent\begin{tabular}{@{}p{.44\linewidth}@{\hspace{1.3em}}p{.48\linewidth}@{}}
$\displaystyle e_{01}=c_{01}$
&
$\displaystyle e_{11}=g_{11}$
\end{tabular}\par}\vspace{4pt}
\begin{flalign}
e_{21}&=
\frac{V\mu_2}{64h^4}
+\frac{V^3\mu_3}{144U^4}
+O(\mu^2)&&\notag\\
e_{31}&=
-\frac{2b^3h^3+9b^2h^2+15bh+9}{72h^3U^3}
-\frac{b(bh+3)\mu_2}{144h^4U}&&\notag\\
&\quad+
\frac{(2b^2h^2+6bh+9)V^2\mu_3}{324hU^5}
+O(\mu^2).&&\label{eq:appB:e}
\end{flalign}

\noindent\textit{Coefficients $h_{ij}$ (Coefficients of \eqref{eq:ns-step4}).}\par\nobreak\vspace{1pt}
{\renewcommand{\arraystretch}{1.25}
\noindent\begin{tabular}{@{}p{.44\linewidth}@{\hspace{1.3em}}p{.48\linewidth}@{}}
$\displaystyle h_{00}=\frac{U\mu_1}{16h^3}+O(\mu^2)$
&
\end{tabular}\par}\vspace{4pt}
\begin{flalign}
h_{20}&=
-\frac{9U^3}{4h^2V^2}
+\frac{(7b^3h^3+22b^2h^2+24bh+9)\mu_1}{16h^3V^2}
-\frac{3U\mu_3}{V}
+O(\mu^2).&&\label{eq:appB:h}
\end{flalign}

\noindent\textit{Coefficients $b_{ij}$ (Coefficients of \eqref{eq:ns-step4}).}\par\nobreak\vspace{1pt}
{\renewcommand{\arraystretch}{1.25}
\noindent\begin{tabular}{@{}p{.44\linewidth}@{\hspace{1.3em}}p{.48\linewidth}@{}}
$\displaystyle b_{01}=
\frac{V(4bh+3)\mu_1}{72hU^2}+c_{01}+O(\mu^2)$
&
\end{tabular}\par}\vspace{4pt}
\begin{flalign}
b_{21}&=
\frac{P_6\mu_1}{1152U^2h^3V^2}
+\frac{9U^2\mu_2}{16h^2V}
+\frac{P_7\mu_3}{12U^2V}
+O(\mu^2)&&\notag\\
b_{31}&=
-\frac{3b^2h^2+9bh+9}{V^2}
+\frac{(bh+\frac{3}{4})P_8\mu_1}{4U^2h^3V^3}&&\notag\\
&\quad
-\frac{3b(bh+3)U^2\mu_2}{2hV^3}
-\frac{2(7b^2h^2+21bh+18)h^2\mu_3}{3VU^2}
+O(\mu^2),&&\notag\\
b_{41}&=
\frac{P_9+(P_4+P_{10})\mu_1}{2048U^2h^3V^4}&&\notag\\
&\quad
-\frac{81U^2(8b^3h^3+36b^2h^2+118bh+75)\mu_2}{256h^2V^4}
-\frac{3P_{11}\mu_3}{64U^2V^3}.
&&\label{eq:appB:b}
\end{flalign}\\

\medskip\noindent\textbf{Polynomial abbreviations (used in \eqref{eq:appB:g} and \eqref{eq:appB:b})}\par\nobreak\vspace{1pt}
\begingroup\footnotesize
\begin{flalign}
P_1={}&
\left(
1-\frac{16}{9}b^2h^4
+(b^3-\frac{16}{3}b)h^3
+(3b^2-4)h^2+3bh
\right)&&\notag\\[-2pt]
&\quad\times
\left(b^2h^2+\frac{15}{8}bh+\frac{27}{16}\right),&&\notag\\
P_2={}&
18b^3h^3-20b^2h^4+54b^2h^2-60bh^3
+54bh-45h^2+18,&&\notag\\
P_3={}&
\left(
9b^3h^3-8b^2h^4+27b^2h^2
-24bh^3+27bh-18h^2+9
\right)&&\notag\\[-2pt]
&\quad\times(4bh+3)(2b^2h^2+6bh+9),&&\notag\\
P_4={}&
288b^7h^7-1536b^6h^8+1944b^6h^6
-11904b^5h^7+5508b^5h^5&&\notag\\
&-40256b^4h^6+13734b^4h^4
-79296b^3h^5+29349b^3h^3&&\notag\\
&-98496b^2h^4+37611b^2h^2
-72360bh^3+24219bh-23652h^2+6075,&&\notag\\
P_5={}&
432b^5h^5-368b^4h^6+2592b^4h^4
-2208b^3h^5+6804b^3h^3&&\notag\\
&-4776b^2h^4+9180b^2h^2-5832bh^3
+6156bh-3591h^2+1620,&&\notag\\
P_6={}&
-144b^5h^5-702b^4h^4-1485b^3h^3
-1683b^2h^2-999bh-243&&\notag\\
&-896b^4h^6-4800b^3h^5-10080b^2h^4
-9288bh^3-2916h^2,&&\notag\\
P_7={}&
18-20b^2h^4+(18b^3-60b)h^3
+(54b^2-45)h^2+54bh,&&\notag\\
P_8={}&
b^5h^5+6b^4h^4+\frac{33}{2}b^3h^3
+\frac{47}{2}b^2h^2+\frac{33}{2}bh+\frac{9}{2}&&\notag\\
&+\frac{40}{9}b^4h^6+\frac{80}{3}b^3h^5
+54b^2h^4+30bh^3-9h^2,&&\notag\\
P_9={}&
12288b^6h^9+98304b^5h^8+316416b^4h^7&&\notag\\
&+497664b^3h^6+377856b^2h^5+110592bh^4,&&\notag\\
P_{10}={}&
640h^2(2bh+3)(4bh+3)
(2b^4h^4+12b^3h^3+36b^2h^2+54bh+27),&&\notag\\
P_{11}={}&
144b^5h^5-464b^4h^6+864b^4h^4
-2784b^3h^5&&\notag\\
&+2268b^3h^3-5944b^2h^4+3060b^2h^2
-5016bh^3&&\notag\\
&+2052bh-1197h^2+540.
&&\label{eq:appB:polynomials}
\end{flalign}
\endgroup

\setcounter{equation}{\value{appBSavedEquation}}
\endgroup

\bibliographystyle{ieeetr}
\bibliography{citation}

\end{document}